\documentclass[aps,prd,reprint]{revtex4-1}
\usepackage{amsmath,amsfonts,amssymb,amsthm,graphicx,enumerate,times,color,bm}
\usepackage{graphicx}
\usepackage{tikz}
\usepackage{tabularx}
\usepackage[utf8]{inputenc}

\newtheorem{theorem}{Theorem}

\newtheorem{lemma}[theorem]{Lemma}

\newtheorem{definition}[theorem]{Definition}

\usepackage{hyperref}

\newcommand{\e}{{\rm e}}

\newcommand{\RR}{\mathbb{R}}

\newcommand{\N}{\mathcal{N}}

\newcommand{\id}{\mathbb{1}}
\newcommand{\mc}[1]{\mathcal{#1}}

\newcommand{\ketbra}[2]{| #1 \rangle \langle #2 |}

\newcommand{\proj}[1]{\vert #1\rangle\!\langle#1 \vert}

\def\id{{\mathbb I}}
\newcommand{\norm}[1]{\left\Vert #1 \right\Vert}

\def\f{{\rm f}}

\newcommand{\Tr}{\operatorname{tr}}
\newcommand{\tr}{\Tr}

\def\v{\mathbf{v}}

\def\b{\bm{\beta}}
\def\a{\bm{\alpha}}

\newcommand{\poly}[1]{\mathrm{poly}\left(#1\right)}
\newcommand{\x}{\mathbf{x}}

\begin{document}
\newcommand{\fu}{Dahlem Center for Complex Quantum Systems, Freie Universit{\"a}t Berlin, 14195 Berlin, Germany}
\author{Paul Boes}
 \author{Henrik Wilming}
 \author{Jens Eisert}
 \author{Rodrigo Gallego}

\affiliation{\fu}

\title{Statistical  ensembles without typicality}

\date{\today}

\begin{abstract}
Maximum-entropy ensembles are key primitives in statistical mechanics from which thermodynamic properties can be derived.
Over the decades, several approaches have been put forward in order to justify from minimal assumptions the use of these ensembles in statistical descriptions. However, there is still no full consensus on the precise reasoning justifying the use of
such ensembles.
In this work, we provide a new approach to derive maximum-entropy ensembles taking a strictly
operational perspective. We
investigate the set of possible transitions that a system can undergo together with an environment, when one only has partial information about both
the system and its environment. The set of all these allowed transitions encodes thermodynamic laws and limitations on thermodynamic tasks as particular cases. Our main result is that the set of allowed transitions coincides with the one possible if both system and
environment were assigned the maximum entropy state compatible with the partial information. This justifies the overwhelming success of such ensembles and provides a derivation without relying on considerations of typicality or information-theoretic measures.
\end{abstract}

\maketitle

Maximum-entropy ensembles, such as the microcanonical or the canonical ensemble, are the pillars on which statistical mechanics
rests. Given some partial information about a system, a vast set of  predictions about its behaviour can be derived by assigning to the system that statistical ensemble which maximizes the entropy compatible with the partial information.
Yet, in some ways this assignment may be seen as being  peculiar in that there exist many other possible physical states that are compatible with this information. 

The assignment of maximum-entropy ensembles is primarily justified by its undoubtable empirical success when it comes to an agreement with
experiment and observation. Thus, unsurprisingly, there has been much work aiming at providing theoretical grounds which explain its empirical success, going back to seminal work by Gibbs \cite{Gibbs1902Elementary}. The most successful general arguments justifying the use of ensembles -- both for classical and quantum systems -- are either based on specific assumptions of the microscopic interactions from which ergodicity can be derived (see Refs.\ \cite{Uffink,Haar} for a review on this approach and its conceptual problems), or based on the notion of typicality. The latter is the observation that the volume of pure quantum states (compatible with the information) that behave like a maximum-entropy ensemble is close to unity, with respect to a relevant measure on state space \cite{Goldstein2006Canonical,Tasaki,Popescu2006Entanglement}.
In these approaches, partially motivated by efforts in quantum thermodynamics \cite{XuerebReview,ThermoReview},
the aim is to show that the system at hand behaves like the ensemble in the precise sense that it will output the same measurement statistics for a restricted, but most realistic and relevant, set of observables. In this way, the agreement between experiments and the assignment of ensembles is justified, with the only notorious problem that the measure that produces the typicality is difficult to justify dynamically. There have been attempts to derive precisely the emergence of canonical ensembles
for most times from microscopic dynamical laws for common locally interacting quantum systems
(for reviews, see Refs.\ \cite{1408.5148,ngupta_Silva_Vengalattore_2011,christian_review}). However, it seems fair to say
that it is still not fully clear yet why the probability of a system being, at any (or most) times,
in a state should be described by this measure -- a state of affairs particularly significant in light of the
importance of this ensemble.

In this work, we provide a very different justification for the use of such ensembles. In contrast to the approaches mentioned before, our aim is not to derive that system's measurement statistics mimic those of the ensemble. Instead, we look at the possible state transitions that can be induced on a system from which one has only partial information (see also Refs.\ \cite{Rio2015,Kramer2016}). More precisely, we consider an initial system described only by partial information in the form of the expectation value of a set of observables. We pose the problem of finding the set of transitions that this initial system can undergo by evolving jointly with an environment when the state of this environment is itself known only partially, that is, up to expectation values with respect to a set of observables that correspond to those of the system. The environment plays the role of a usual heat bath and the set of transitions encode any possible task: extracting work, reaching a colder/warmer state, performing a computation or any other. Our main result is that, for \emph{any} initial state, the possible state transitions on such a system under partial information coincide exactly with those possible if the system and the environment were initially in the maximum-entropy ensemble state compatible with the partial information. This then not only justifies the use of the canonical ensemble to represent a system under partial information, it also allows one to derive the building blocks of phenomenological thermodynamics without assuming systems to be represented by this ensemble. In fact our results can be seen as a derivation of the Gibbs entropy and the Clausius inequality without a priori assigning equilibrium states to the systems involved.
Finally, since our results hold for any initial state, they do not suffer from the problem of typicality approaches mentioned above and allow us to avoid assumptions about the system's Hilbert-space dimension (apart from being finite). In particular, our results also hold for small, individual quantum systems.

\section{Motivating example}\label{sec:example}

\begin{figure*}[ht!b]
\includegraphics{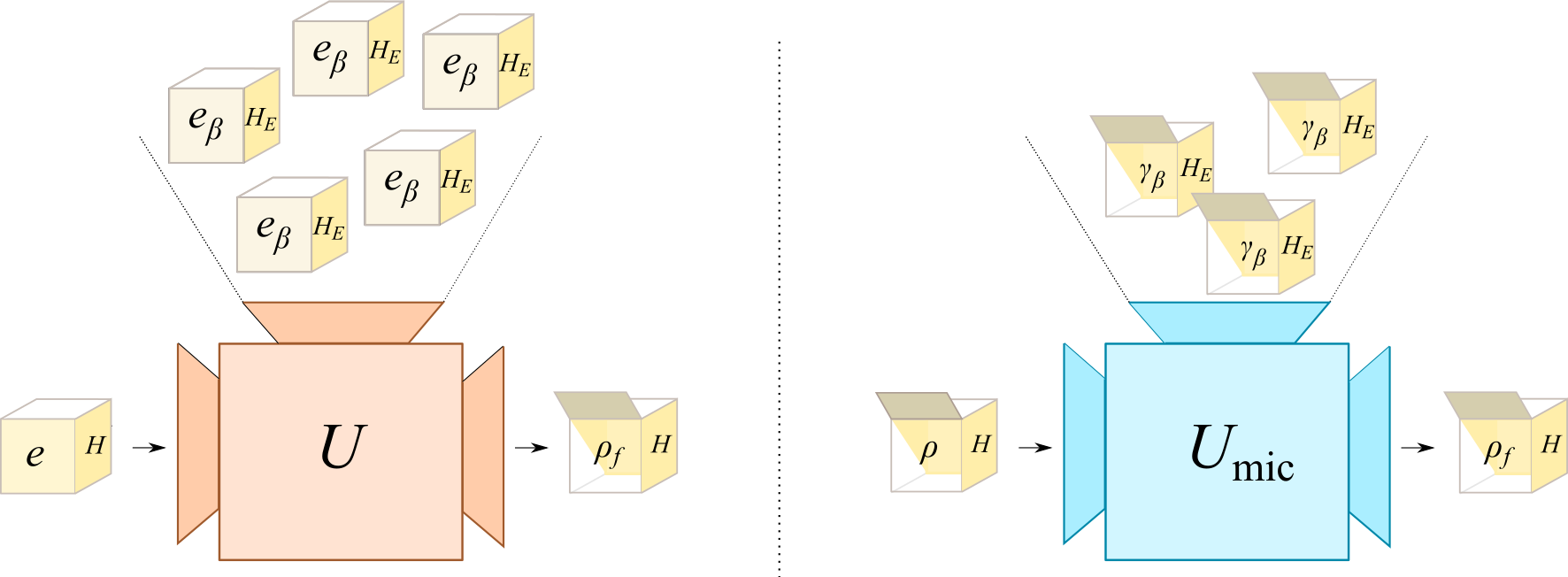}
\caption{Pictorical representation of the equivalence between scenarios a) with partial information (macrostate operations) and b) scenarios where systems are assigned the corresponding ensemble (microstate operations). Closed boxes represent systems from which we only know some partial information, in this case the mean energy. Inside the box there is the actual microstate unknown to us if the box is closed. Scenario a) shows the situation where one has an initial system of which only the mean energy $e$ is known and one can use any environment, being again limited to knowledge of its initial average energy $e_{\beta}$. The question is whether we can find a unitary $U$ that takes the two systems, regardless of what is actually inside of them, to one box for which we are certain that we will find inside the microstate $\rho_f$. The answer to this question is provided by scenario b), where the initial boxes of system and environment are both open (implying that we know what is the microstate) and populated with the maximum-entropy ensemble. $U$ exists if and only if there exists a unitary $U_{\mathrm{mic}}$ that implements the transition in b) when taking $\rho=\gamma_{e}(H)$. This shows that a thermodynamic transition is possible if and only if it is also possible under the assignment of ensembles to systems.}
\label{fig:operationaldepiction}
\end{figure*}

	We begin the presentation of our setting with a motivating example. Consider a small quantum system $S$ with Hamiltonian $H$ within an environment $E$ at temperature $T$ and with Hamiltonian $H_E$, that is, an environment in the canonical ensemble at that temperature and Hamilonian.

Given an initial quantum state $\rho \in D(\mc H)$ of the system, we can ask which final states of the system can be reached by coupling the system to the environment and evolving the joint system $SE$ in such a way that the global entropy and energy remain unchanged, if one assumes perfect control over both the environment Hamiltonian $H_E$ and the coupling, but for a fixed temperature $T$. Naturally, the answer to this question will strongly depend on the particular initial quantum state of $S$. For instance, the maximally mixed state $\rho=\id_{S}$ and an energy eigenstate $\rho' = \proj{E_i}$ will generally allow for very different state transitions. That is, there will exist some final state $\rho_f$ that can be reached by some entropy and energy preserving procedure $O$ from $\rho'$, while no such procedure exists for $\rho$. Call this scenario the microstate scenario, because here one has full information about the actual ``microstates'' -- \emph{i.e.} quantum states -- of the system and the environment.

	Suppose now that, instead of knowing the exact state of the system, one initially only knows its mean energy to be $e$ with respect to $H$. We capture this partial information in what we call a \emph{macrostate} of the form $(e,H)$. In this case, one can again ask which are the reachable states given that partial information. However, in this case the difficulty is that, in general, there will be many microstates compatible with this information. For instance, suppose that $(e,H)$ is compatible with both $\rho$ and $\rho'$. In this case $\rho_f$ cannot be reached anymore because there is at least one state -- $\rho$ in the previous example -- compatible with the initial information for which $\rho_f$ is unattainable. That said, one concludes that in order to reach some final state $\rho_f$, if only partial information about the initial state of $S$ is had,  one requires a \emph{single} operational procedure $O$ that takes $\emph{any}$ state compatible with the initial information to $\rho_f$. Note that this scenario is undesirably  asymmetric in that the system's state is represented by a macrostate $(e,H)$ (capturing our partial knowledge), while the environment microstate is fully known to be in the canonical ensemble at temperature $T$. Hence, one can go one step further and consider a situation in which not only does one only know the system's initial mean energy, but also the environment is described by a macrostate $(e_E,H_E)$. In this case, it becomes even more difficult to reach a given final microstate $\rho_f$, since now there has to exist a \emph{single} procedure $O$ that prepares $\rho_f$ from \emph{any} microstate of $S$ compatible with $e$ and \emph{any} environment microstate compatible with $e_E$. Indeed, it may seem that in general no transition is possible under these circumstances. At the same time, this scenario most accurately describes the situation that one in fact faces in phenomenological thermodynamics, where only coarse-grained information is had about both system and environment. Call this last scenario then the macrostate scenario, because here both system and environment are described by macrostates $(e,H)$ and $(e_E,H_E)$ respectively.

	The main result of this work is to show that, not only do there exist possible transitions in the macrostate scenario, moreover these transitions are fully characterized by assigning maximum-entropy ensembles to the macrostates involved: Under a natural model of operational procedures modelling thermodynamic transitions that we introduce below, given some value $e$, a final microstate $\rho_f$ can be reached in the macrostate scenario if and only if it can be reached in the microstate scenario from the canonical ensemble state of energy $e$. Since the canonical ensemble is moreover the only state for which this equivalence holds, this result provides an explanation for the important role that the canonical ensemble plays in statistical mechanics, a theory formulated in the microstate scenario, to describe phenomenological thermodynamics, a theory formulated in the macrostate scenario.
	\section{Setting and Results} 
\label{sec:results}

We now proceed to make the notion of the microstate- and macrostate scenario rigorous and introduce our model of thermodynamic transitions, \emph{i.e.} the transitions that a system $S$ can undergo together with an arbitrary environment at fixed temperature.

Consider a $d$-dimensional quantum system $S$ whose mean energy with respect to the Hamiltonian $H$ is known to be $e$. We refer to the pair $(e,H)$ as the ``macrostate'' of the system, as it corresponds to a state of coarse-grained information about the system. Note, however, that we do not assume that the system is macro\emph{scopic}, i.e. that $d \gg 1$. Every macrostate of the system corresponds to an equivalence class $[e]_H$ of ``microstates'' $\rho \in D(\mc H)$ of the system, namely all those density matrices whose mean energy with respect to $H$ is $e$, with $\mc E(\rho):= \Tr(\rho H)= e$. The canonical ensemble corresponding to a macrostate $(e,H)$ is then
\begin{align}\label{eq:canonicalensemble}
	\gamma_e(H) &:= \frac{e^{- \beta_S(e) H}}{\tr(e^{-\beta_S(e) H}) },
\end{align}
where $\beta_{S}(e)$ is chosen such that $\tr(\gamma_{e}(H) H)=e$. Note that, by construction, $\gamma_{e}$ is the maximum-entropy element in ${[e]}_{H}$ and exists for every macrostate.
As is clear from the example, in the following, we will often be concerned with making comparative statements about the microstate- and the macrostate scenarios. To simplify the presentation and highlight similarities between these scenarios, we now introduce the following convention: Let $M$ be any map acting on microstates. Then $M((e,H)):= M([e]_H)$ is the corresponding macrostate-level map. This notation will prove convenient in several ways. For instance, the requirement that an operation $O$ maps all the states $\rho$ compatible with $(e,H)$ into the state $\rho_f$ is simply expressed by
\begin{align}
O((e,H))=\rho_f.
\end{align}
Similarly, this notation can be also used to express operations on tensor products of macrostates. For instance, the expression \begin{align}\label{eq:reachable}
O( (e,H) \otimes (e_E,H_E))=\rho_f
\end{align}
implies that $O( \rho \otimes\rho_E)=\rho_f$ for all $\rho$ and $\rho_E$ compatible with $(e,H)$ and $(e_E,H_E)$ respectively.

\subsection{Thermodynamic operations on macrostates}\label{sec:definitionsmain}

Let us now describe and justify more precisely the form of a general macrostate operation as informally described in Section \ref{sec:example}. With these operations we aim at capturing in full generality \emph{any} possible transition that a system can undergo together with a heat bath. Hence, in order to describe an arbitrary macrostate operation, one is perfectly free to choose as an environment \emph{any} system of arbitrary Hilbert space dimension and with an arbitrary Hamiltonian $H_E$. As mentioned before, we do not assume that $E$ is in a canonical ensemble -- which would be fully determined by the inverse temperature $\beta:=(k_B T)^{-1}$, dimension and Hamiltonian -- but to have a partial description in terms of its average energy, thus assigning to it a macrostate $(e_E,H_E)$. We assume, as it is standard when considering thermodynamic operations
\cite{Horodecki2011Fundamental,Brandao2015Second,Brandao2013Resource}, that the system and the environment are initially uncorrelated, hence one initially possesses the macrostate compound $(e,H)\otimes (e_E,H_E)$. Naturally, the attachment of an uncorrelated environment can be iterated an arbitrary number of times, say $N$, bringing each time a new environment with an arbitrary dimension and Hamiltonian.

Moreover, since the macrostates provided by the environment model a bath, it is natural to assume that there exists a functional relationship between the environment Hamiltonian and the energy. In particular, we will assume this relationship to be that $e_E=e_{\beta}(H_E)$, where
\begin{equation}\label{eq:thermalenergy}
 e_{\beta}(H_{E}):=\tr\left( \frac{e^{-\beta  H_{E}}}{\tr(e^{-\beta H_{E}}) }H_{E}\right)
\end{equation}
is the thermal energy of a bath at inverse temperature $\beta$ and with Hamiltonian $H_E$. This assumption will be further discussed below. Dropping further the dependence on the Hamiltonian in \eqref{eq:thermalenergy} when it is clear from the context, the most general form of an initial macrostate then is of the form

\begin{equation}\label{eq:initial_macro}
(e,H) \bigotimes_{i=1}^{N} (e_\beta,H_{E^i}).
\end{equation}
Given this model of the environment, we now turn to the describing the
model of the joint evolution. Here,
we aim at modeling the isolated evolution of $SE$, in the sense that it preserves the energy and entropy of the compound. Regarding the energy, one has to take into account that only mean values of the energy are accessible, hence it is most reasonable to impose only that the mean energy is preserved \cite{Skrzypczyk2014work,Guryanova2016GGE,Halpern2015Beyond}, while noting that the mean energy must be preserved for all the initial microstates compatible with our initial macrostate \eqref{eq:initial_macro} (see Section \ref{sec:breakdown} for a thorough discussion and possible alternatives). Regarding entropy conservation, we enforce it by imposing a unitary evolution of the compound. We note, however, that our results also hold for larger set of operations such as probabilistic mixtures of unitaries or entropy non-decreasing operations, or even more generally, any set of operations that contains unitary evolutions as a particular case.

Let us now, for sake of clarity, enumerate the assumptions that come into play when describing macrostate operations:

\begin{itemize}
\item[i)] \textbf{Thermal energy environments:} Given an environment with Hamiltonian $H_E$, then the associated macrostate is given by $(e_{\beta}(H_E),H_E)$, where $e_{\beta}(H_E)$ is the thermal energy at reference temperature $T$.
\item[ii)] \textbf{Uncorrelated subsystems:} One can incorporate environments that are initially uncorrelated with the initial system.
\item[iii)] \textbf{Unitary evolution:} The compound $SE$ undergoes a unitary evolution.
\item[iv)] \textbf{Global mean energy conservation:} The unitary evolution of $SE$ is such its mean energy is preserved \emph{for all} the states (both of $S$ and $E$) compatible with our partial information.
\end{itemize}
Before turning to the formal definition of macrostate operations on the basis of these assumption, let us briefly comment on the assumption that environment macrostates have thermal energy \eqref{eq:thermalenergy}. Clearly, this amounts to assume that environment macrostates have the same mean energy as the canonical ensemble at inverse temperature $\beta>0$,
\begin{align}\label{eq:canonicalensemble2}
	\gamma_{\beta}(H_{E}) &:= \frac{e^{-\beta H_{E}}}{\tr(e^{-\beta H_{E}}) },
\end{align}
where we make the convenient abuse of notation of writing $\beta$ directly as the subindex, unlike \eqref{eq:canonicalensemble} where the mean energy was used instead \footnote{This is indeed unproblematic since $e$ and $\beta$ are in one to one correspondence, hence we will use $\beta$ or $e$ indistinctively when it is clear from the context.}.

We emphasize that \eqref{eq:thermalenergy} does not amount to assuming that the environment \emph{is} in the canonical ensemble -- which would beg the question by giving a prominent role to the canonical ensemble -- since many states other than the canonical ensemble fulfilling \eqref{eq:thermalenergy} exist. Nevertheless, assumption i) could raise the criticism that our further results -- the justification of ensembles -- rely on a seemingly arbitrary energy assignment for the macrostate of $E$, as given by \eqref{eq:thermalenergy}. However, we show in Appendix \ref{sec:trivializing} that \eqref{eq:thermalenergy} is the only possible assignment so that macrostate operations reflect indispensable features of thermodynamical operations. More precisely, we prove that \eqref{eq:thermalenergy} is the only energy function that does not allow one to extract an arbitrary amount of work from $E$ alone -- even if only partial information is given. Even more dramatically, it is the only energy function that does not trivialize macrostate operations, in the sense that any possible transition would be possible. Hence, \eqref{eq:thermalenergy} can be regarded as a necessary feature of an environment so that thermodynamic operations are sensibly accounted for in the formalism.

Finally, combining the notational convention for operations on macrostates, assumptions i-iv), and denoting the global mean energy as $
\mc{E}(\rho_{SE}):=\tr(\rho_{SE}H_{SE})$,
we define formally the set of macrostate operations with an environment at inverse temperature $\beta$:
\begin{definition}[Macrostate operations]\label{dfn:macro}
We say that $\rho_f$ can be reached by \emph{macrostate operations} from $(e,H)$, which we denote by
\begin{equation}
(e,H) \overset{\beta\text{-mac}}{\longrightarrow} \rho_f,
\end{equation}
if for any $\epsilon>0$ and $\epsilon'>0$ there exists an environment -- that is, a set of $N$ systems with
respective Hamiltonians $H_{E^i}$ -- and a unitary $U$ on $SE$, so that
\begin{equation}\label{eq:def_mac_op_gge}
 \rho_f \approx_\epsilon \tr_E\left( U \: (e,H) \bigotimes_{i=1}^{N} (e_{\beta},H_{E^i}) \: U^{\dagger}\right)
\end{equation}
while preserving the overall mean energy
\begin{equation}
\label{eq:energyconservation}
\mc{E}\bigg( U \: (e,H) \bigotimes_{i=1}^{N} (e_{\beta},H_{E^i}) \: U^{\dagger}\bigg)\approx_{\epsilon'}\mc{E}\bigg( (e,H) \bigotimes_{i=1}^{N} (e_{\beta},H_{E^i}) \bigg).
\end{equation}
\end{definition}

Here, we use $\approx_\epsilon$ to say that two quantities differ by at most $\epsilon$ in trace-norm,
or in absolute value for expectation values. Note that although we allow for errors $\epsilon,\epsilon'$ in the transition and in the mean energy conservation, those errors can be made arbitrarily small, hence it is for all practical purposes indistinguishable from an exact transition with exact mean energy conservation.
It also is important to stress again that, in the previous definition and following the notation introduced with Eq.\ \eqref{eq:reachable}, both \eqref{eq:def_mac_op_gge} and \eqref{eq:energyconservation} have to be fulfilled for all the microstates compatible with the macrostates appearing in those equations.
See Fig. \ref{fig:operationaldepiction} a) for a schematic description of macrostate operations as presented in Definition \ref{dfn:macro}.

\subsection{Thermodynamic operations on microstates and main result}
As stated before, our main result consists in showing that not only is the set of reachable microstates under macrostate operations in general non-empty, it can also be characterized exactly by the corresponding canonical ensembles. In order to be able to state this correspondence between macrostates and their canonical ensembles formally, we will now introduce \emph{microstate operations} as the corresponding model of thermodynamic transitions in the microstate scenario. These differ from macrostate operations only in that we assign a particular microstate to $S$ and $E$. In other words, microstate operations are the complete analogue of the operations in Definition\ \ref{dfn:macro}, but with full information about the actual quantum states involved. Hence, the conditions \eqref{eq:def_mac_op_gge} and \eqref{eq:energyconservation} are modified, for microstate operations, in that they have to be fulfilled for a single state and not for a set of states compatible with our knowledge.

\begin{definition}[Microstate operations]\label{dfn:micro}
We say that $\rho_f$ can be reached by \emph{microstate operations} from $\rho$, which we denote by
\begin{equation}
\rho \overset{\beta\text{-mic}}{\longrightarrow} \rho_f,
\end{equation}
if for any $\epsilon>0$ and $\epsilon'>0$ there exists an environment --that is, a set of $N$ systems with Hamiltonians $H_{E^i}$-- and a unitary $U$ on $SE$, so that
\begin{equation}
 \rho_f \approx_\epsilon \tr_E\left( U \: \rho \bigotimes_{i=1}^{N} \gamma_{\beta}(H_{E^i}) \: U^{\dagger}\right)
\end{equation}
while preserving the overall mean energy
\begin{equation}
\label{eq:energyconservationmicro}
\mc{E}\bigg( U \: \rho \bigotimes_{i=1}^{N} \gamma_{\beta}(H_{E^i}) \: U^{\dagger}\bigg)\approx_{\epsilon'}\mc{E}\bigg( \rho \bigotimes_{i=1}^{N}\gamma_{\beta}(H_{E^i}) \bigg).
\end{equation}
\end{definition}
An operationally inspired illustration of the two types of operations as well as of our result is provided in Fig.\  \ref{fig:operationaldepiction}.

\begin{figure*}[htb]
  \includegraphics{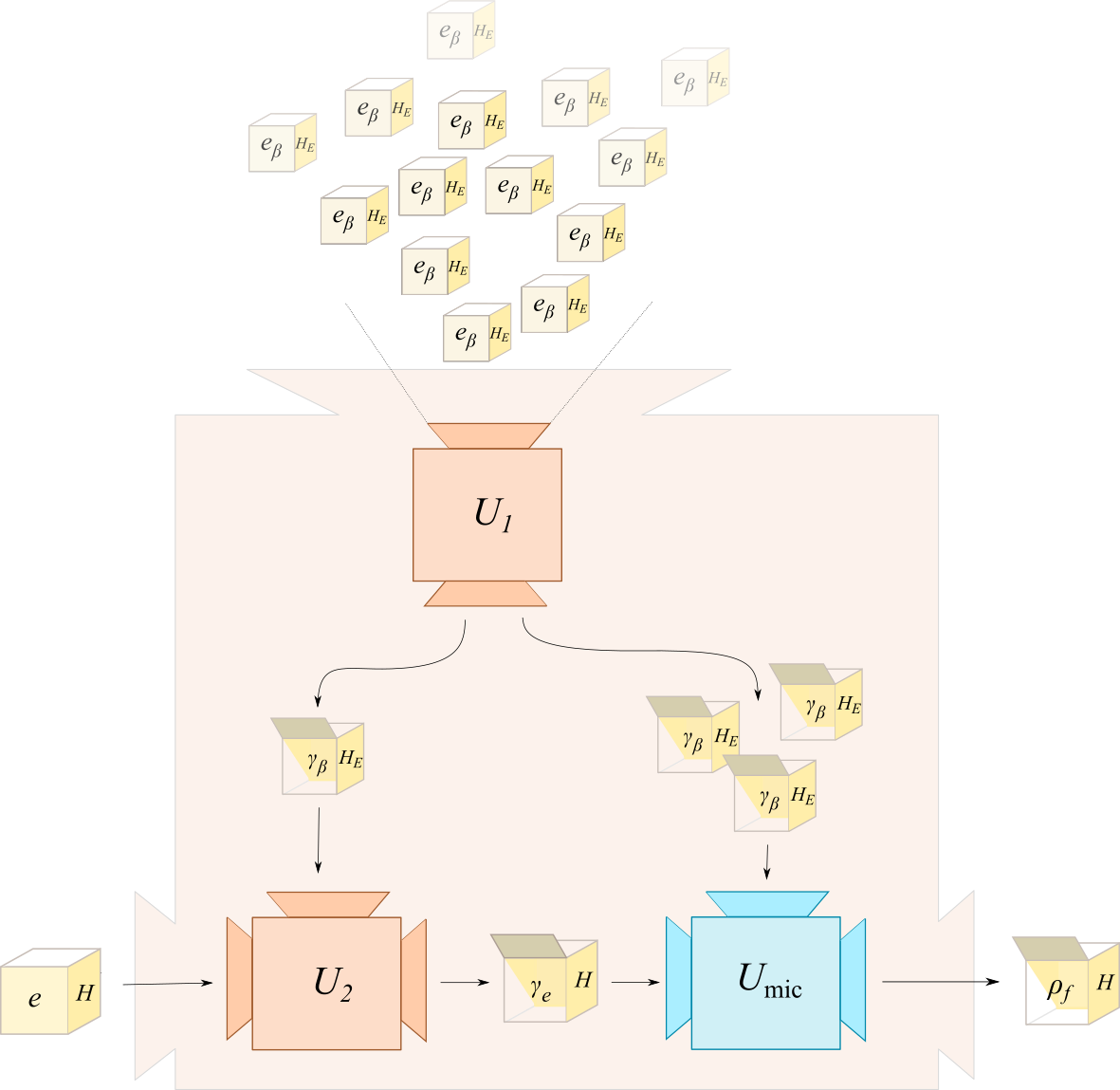}
  \caption{The figure sketches the proof of our main result. More particularly, we show how an operation of the form of Fig.\  \ref{fig:operationaldepiction} ii) can be used to build an operation of the form Fig.\  \ref{fig:operationaldepiction} i). This gives the direction $\Leftarrow$ in \eqref{eq:equiv} for the equivalence of Theorem \ref{thm:main} (the other direction is trivial, see Appendix \ref{app:main}). The construction has three sub-blocks: \textbf{Box ${\bf U_1}$} represents the fact that one can obtain the microstate $\gamma_\beta(H_E)$ to arbitrary precision from many copies of the macrostate $(e_\beta(H_E),H_E)$ using a macrostate operation (interestingly, this can be done with exact energy conservation). This result relies on a central limit theorem and typicality results for individual energy eigenspaces of many non-interacting systems. \textbf{Box ${\bf U_2}$} operates by choosing as $H_E$ as a rescaled version of $H$ and showing that one can then obtain the microstate $\gamma_{e}(H)$ using a macrostate operation. \textbf{Box ${\bf U_\text{mic}}$} exists by assumption: it uses the microstate operation to obtain $\rho_f$ from $\gamma_{e}(H)$ (it is the one represented in Fig.\  \ref{fig:operationaldepiction} ii)).}
  \label{fig:proof}
\end{figure*}

In this setup, we call a macrostate $(e,H)$ and a microstate $\rho$ \emph{operationally equivalent}, denoted as $(e,H) \sim_\beta \rho$, if
\begin{equation}
	(e,H) \overset{\beta\text{-mac}}{\to} \rho_f \Leftrightarrow \rho \overset{\beta\text{-mic}}{\to} \rho_f. \label{eq:equiv}
\end{equation}
Whenever a macrostate and a microstate are related by the equivalence $\sim_\beta$, then, concerning the possible thermodynamic transitions, they are equivalent descriptions of the system. We are now in a position to state our main result.

	 \begin{theorem}[Equivalence with the canonical ensemble] \label{thm:main}
	For any $\beta \neq 0$, the macrostate $(e,H)$ is operationally equivalent to the corresponding canonical ensemble compatible with the partial information $e$. That is,
	 	\begin{align}
	 	    (e,H) &\sim_\beta \gamma_{e} (H).
	 	\end{align}
	 \end{theorem}

This theorem shows that, whenever the behaviour of a system under partial information concerns the possible thermodynamic transitions, a macrostate can be treated as if it was in its corresponding canonical ensemble, in the sense that they their behaviours coincide exactly. A sketch of the proof, for illustration of the idea, is given in Fig.~\ref{fig:proof}. The full proof appears in Appendix \ref{app:main}.

Lastly, let us note that all of the above, including the operations and the notion of operational equivalence, can straightforwardly be generalised to the more general case of a set $\mc Q =\{Q^j\}$ of $n$ commuting observables replacing $H$, a vector $\v$ of expectation values for each observable replacing $e$ and by now parametrising the environment by a vector of inverse ``temperatures'' $\b=(\beta^1,\ldots,\beta^n)$ encoding other intensive quantities. In this case, we obtain an operational equivalence between the macrostate $(\v,\mc{Q})$ and the corresponding maximum-entropy ensemble compatible with the partial information. More precisely, we obtain that, as long as $\beta^j \neq 0$ for all $j$,
\begin{align}
	 	    (\mathbf{v},\mc Q) &\sim_{\b} \gamma_{\v}(\mc Q),
\end{align}
where, in exact analogy to \eqref{eq:canonicalensemble}, $\gamma_{\v}(\mc Q)$ is the so-called
\emph{generalised Gibbs ensemble (GGE)}
\cite{Jaynes1957Information,christian_review,Rigol_etal08,Llobet2015Work,Halpern2016Microcanonical,Lostaglio2017}
\begin{equation}\label{eq:ggedef}
\gamma_{\v}(\mc Q):=\frac{e^{-\sum_j \beta^j_S(\v) Q^j}}{ \tr\left(e^{-\sum_j \beta^j_S(\v) Q^j}\right) } ,
\end{equation}
with $\beta^j_S(\v)$ being functions such that $\tr(Q^j \gamma_{\v}(\mc Q)) = v^j$.
The scenario and derivation is completely analogous to that yielding Theorem \ref{thm:main} and it is presented in Appendix \ref{sec:app:gge}.

\section{Discussion} 
\label{sec:discussion}

At a conceptual level, we regard as our main contribution the theoretical justification, from an operational perspective, for the common and empirically
extraordinarily well-supported use of the canonical ensembles in thermodynamics to describe systems in settings of partial information. The key step in this justification has been to prove a coincidence in behaviour with respect to thermodynamic transitions. The relevance of this coincidence is that many thermodynamic tasks and the laws of thermodynamics can ultimately be formulated as reflecting state transitions. To illustrate this statement,
we will now discuss how the equivalence on reachable states can be used to derive, as particular cases, common situations in thermodynamics such as work extraction, or more generally, quantitative versions of the second law of thermodynamics.

\subsection{Thermodynamic tasks as macrostate operations}
\label{sec:thermodynamic_tasks}

Let us consider the following task: One is given a system $S$ from which only the Hamiltonian and its mean energy $e$ are given. For instance, $S$ might be a burning fuel which one wants to use in a heat engine to perform work together with an environment. This common scenario is tackled in phenomenological thermodynamics by assigning to the system a temperature $T_S$ and to the environment a temperature $T$. The optimal amount of work that can be performed is simply given by the difference of free energies of $S$ during the process. Note that phenomenological thermodynamics operates at a level where only partial information --the thermodynamic variables -- are given about both the system and the environment. Furthermore, the operation of such a heat engine is effectively independent of the precise microstate that describes $S$ and $E$, exactly in the same spirit as that of
Definition\ \ref{dfn:macro}.

From the perspective of statistical mechanics, the assignment of a temperature $T_S$ and $T$ is understood as the assumption that both systems are in a canonical ensemble. Indeed, if we assume the system and the environment are initially in the state
\begin{equation}\label{eq:assigning_canonical_work_extraction}
\gamma_{e} \otimes \gamma_{\beta} :=\gamma_{e}(H) \otimes \gamma_{\beta}(H_E)
\end{equation}
one can formally derive limitations on the work $\Delta W$. The problem amounts to finding how much one can reduce the energy of the whole compound by \emph{any} unitary operation that does \emph{not} conserve the energy and assuming that all of the remaining energy can be extracted as work. One then obtains that this value is determined by the free energy as (see, e.g., Ref.\ \cite{Skrzypczyk2014work})
\begin{eqnarray}
\nonumber \Delta W^{\text{opt}}&:=&\max_{U,H_E} \left[\mc E(\gamma_{e} \otimes \gamma_{\beta}) - \mc E(U\gamma_{e} \otimes \gamma_{\beta}U^{\dagger}) \right]\\
\label{eq:maxworkfreeenergy}&= & \Delta \mc E_S - T \Delta \mc S_S:= \Delta \mc F_S,
\end{eqnarray}
where we denote the energy by $\mc E(\rho_{SE})= \tr (\rho_{SE}(H_S+H_E))$, $\Delta \mc E_S$ is the energy difference on $S$ and $\Delta \mc S_{S}$ is the difference of the von Neumann entropy on $S$. This yields the bound in terms of the free energy $\mc F_S=\mc E_S-\beta^{-1}\mc S_S$ of the system and it relies only on the first law of thermodynamics $\Delta \mc E_{SE} =-\Delta W$ and the prescription of canonical ensembles to the system and environment.

We will now show that one can use Theorem \ref{thm:main} to derive the bound \eqref{eq:maxworkfreeenergy} without relying on the assumption \eqref{eq:assigning_canonical_work_extraction} which assigns maximum entropy ensembles to the systems at hand. The system $S$, given the partial information, is described by the macrostate $(e,H)$. We also have at our disposal an environment in any macrostate of the form $\bigotimes (e_{\beta}(H_{E^i}),H_{E^i})$. The goal is to perform work with a protocol in such a way that it achieves this work extraction for all possible microstates in the respective equivalence classes, $[e]_H$ and $[e_{\beta}(H_{E^i})]_{H_{E^i}}$ for all $i$, in a similar way to the way the laws of phenomenological thermodynamics allow one to extract work regardless of the actual microstates of the systems involved. It is clear that
\begin{equation}
\gamma_{e}(H) \overset{\beta\text{-mic}}{\to} \gamma_{e}(H) \:\: \forall \: e,H.
\end{equation}
Hence, by invoking Theorem \ref{thm:main} one has also that
\begin{equation}
\begin{split}\label{eq:workmac}
(e,H) &\overset{\beta\text{-mac}}{\to}\gamma_{e}(H) ,\\
(e_{\beta}(H_E),H_E) &\overset{\beta\text{-mac}}{\to} \gamma_{\beta}(H) .
\end{split}
\end{equation}
Once we have the system $S$ and the environment $E$ in the states of at the r.h.s.\ of  \eqref{eq:workmac}, we simply apply the unitary achieving the maximum in Eq.\ \eqref{eq:maxworkfreeenergy}. In this way an amount of work given by $\Delta \mc F_S$ is extracted. The fact that this is the optimal possible value that works for all microstates in $[e]_H$ is trivial, since the work extraction has to be successfully implemented if the system is given is in the state $\gamma_{e}(H) \in [e]_H$, for which the optimal value is $\Delta \mc F_S$ as given by Eq.\ \eqref{eq:maxworkfreeenergy}.

We conclude then that the optimal work that can be extracted from a system and an environment, from which we only know their mean energy, coincides precisely with the optimal work when system and environment are described by their corresponding canonical ensemble. A completely analogous argument applies to any other conceivable task that can be formulated as concerning state transitions between microstates, both thermodynamically but also, and more generally, tasks with other conserved quantities.

\subsection{The second law of thermodynamics and the Clausius inequality}
Now we show that the second law of thermodynamics can be recovered by using Theorem \ref{thm:main}. More particularly, we show that the set of achievable states $\rho_f$ that can be reached by a transition of the form
\begin{equation}\label{eq:mactransitionlaws}
 (e,H) \overset{\beta\text{-mac}}{\to} \rho_f
\end{equation}
can be determined only by merely taking into account the free energy $\mc F$. First note that by Theorem \ref{thm:main} the set of achievable $\rho_f$  coincides with those that can be achieved by microstate operations of the form
\begin{equation}
 \gamma_{e}(H) \overset{\beta\text{-mic}}{\to} \rho_f .
\end{equation}
The set of achievable states by microstate operations has been investigated in Ref.\  \cite{Skrzypczyk2014work}, where it is shown that the transition is possible if and only if the free energy decreases. Hence, we arrive at the second law of the form
\begin{equation}\label{eq:macsecondlaw}
 (e,H) \overset{\beta\text{-mac}}{\to} \rho_f \Leftrightarrow \mc{F}(\gamma_{e}(H) )\geq \mc F (\rho_f).
\end{equation}
Importantly, this result can also be seen as a derivation of the free energy as a state function $F(e,H)$ on macrostates, by setting $F(e,H) = \mc F(\gamma_{e}(H))$. Since the energy is already naturally defined for macrostates we then also obtain the derived Gibbs entropy
\begin{align}
S(e,H) := T(e - F(e,H)).
\end{align}
Interpreting the change of energy on the system as heat $\Delta Q$, we see that a transition between macrostates using macrostate operations is possible if and only if
\begin{align}
\Delta Q \geq T \Delta S.
\end{align}
We thus find that a state-transition between macrostates is possible if and only if the \emph{Clausius inequality}
is fulfilled.

Lastly, we highlight that a generalisation of the same results for the case of multiple commuting observables is possible combining in a similar fashion Theorem \ref{thm:gge} (App.~\ref{sec:app:gge}) with the results of \cite{Guryanova2016GGE} to arrive at a
formulation of the \emph{second law} of the form
\begin{equation}
 (\v,\mc Q) \overset{\beta\text{-mac}}{\to} \rho_f \Leftrightarrow \mc{G}(\gamma_{\v)}(\mc{Q}) )\geq \mc G (\rho_f)
\end{equation}
where $\mc{G}$ is the so called free entropy defined as
\begin{equation}
\mc{G}(\rho)=\sum_{j} \beta_j \tr (\rho \: Q^j) - \mc S (\rho). \label{eq:free_entropy}
\end{equation}

\subsection{Comparison with existing work}

There exist several complementary approaches to justify the use of or single out maximum-entropy states in thermodynamics. As stated already in the introduction, the novelty of our approach lies in specifically assigning ensembles based on the set of possible thermodynamic transitions. This is in contrast with previous approaches, where canonical ensembles are justified based on measurement statistics of relevant observables. Both perspectives -- the one presented here and previous approaches -- can be fairly incorporated in a more general formulation about what is meant by a justification of the use of ensembles: the representation of a system's state by a statistical ensemble is justified with respect to some property if one can, on reasonable grounds, derive that the ensemble and the state behave exactly the same with respect to this property.
Approaches based on notions of typicality usually consider as system states pure quantum states and the measurement statistics of some restricted set of observables -- often local observables -- as the property to be reproduced by the ensembles \cite{Goldstein2006Canonical,Popescu2006Entanglement}.
In contrast, in the present work, the system states are macrostates of partial information and the property is with respect to achievable state transitions under thermodynamic evolution. Theorem \ref{thm:main} justifies the assignment of maximum-entropy ensembles to macrostates with respect to such transitions. Macrostates are arguably the most common state assignment in thermodynamics,
being at the root of discusssions of the link of statistical mechanics and phenomenological thermodynamics,
in that one often has knowledge of a system's state only up to its expectation values. Hence,
this result provides a very broad operational justification of the use of maximum-entropy ensembles for a plethora of
thermodynamical processes.

Another aspect that distinguishes our approach from other notions based on typicality is that we do not need to introduce a measure on quantum states or make any particular assumption on the dynamics. More precisely, known approaches based on typicality  consider a given subset of quantum states and show that measurement statistics coincide with those of the ensemble for \emph{most} of the quantum states within the subset. However, there is no general argument to advocate that one will find in nature precisely those states for which the statistics resemble those of the ensemble, even though these states comprise the vast majority according to reasonable measures. In contrast, one of the main features of our results is that it works for all and not for most of the quantum states that are compatible with the partial information. First, we demand that the transitions from macrostates, as given abstractly by \eqref{eq:reachable}, reach $\rho_f$ for \emph{all} the states compatible with the partial information. It would be analogous to the notion of typicality if we would instead demand that $\rho_f$ is reached only from \emph{most} of the microstates according to some state measure, but this is actually not required to derive our main results. Secondly, the equivalence between the macrostate and the corresponding ensemble holds for \emph{all} possible macrostates, instead of just for a vast majority of the macrostate according to some measure on the possible values of the partial information. Most importantly, we stress that the equivalence between the macrostate and the ensemble holds \emph{irrespectively of the system's dimension}. To put it in more practical terms, our results imply that a system, even if made of a few qubits, behaves as if it was in its maximum entropy ensemble when it comes to state transitions under joint evolution with a possibly large bath. This is true in a \emph{single-shot} regime -- considering transitions on a single copy of the system at hand -- without having to rely on taking the thermodynamic limit where transitions of large number of copies are considered instead \cite{Sparaciari2016,Bera2017}.

Lastly, it may seem that our approach is closely related to that of the famous \emph{Jaynes' principle} according to which a system should always be assigned the maximum-entropy state consistent with what one knows about it \cite{Jaynes1957Information,Jaynes1957Informationa}. What both approaches have in common is that they consider the question of assigning microstates to macrostates. However, apart from this they differ considerably: Jaynes motivates his principle on the basis of Shannon's findings about the uniqueness of the \emph{Shannon entropy} as an asymptotic measure of information. In contrast, our approach does not require us to assume any privileged measure of information, or even rely on any consideration about information measures at all. Moreover, as noted in the preceding paragraph, our approach also makes no reference to an asymptotic setting. Instead, in our work, we define a task on an individual system and investigate how an experimenter's partial knowledge about the system impacts her ability to execute this task. The canonical ensemble then naturally emerges as an effective representation of the experimenter's operational abilities in this setting. Again, no recourse to a measure of information, average performance, or even a subjectivist account of probabilities is required in our setting.

\subsection{Operational equivalence breaks for exact energy conservation} \label{sec:breakdown}

Theorem \ref{thm:main} establishes the operational equivalence between macrostates and their corresponding maximum-entropy ensembles based, among others, on assumption iv) in Section  \ref{sec:results}, where it is assumed that the mean value of the energy is preserved. In this section, we consider the stronger case in which assumption iv) is replaced by assuming \emph{exact} energy conservation in the following sense:
\begin{itemize}
\item[iv')] The unitary evolution $U$ commutes with the total Hamiltonian,
\begin{equation}\label{eq:exactcommutation}
[U,H_S+H_E]=0.
\end{equation}
\end{itemize}
We define, equivalently to the results of Section  \ref{sec:results}, macrostate and microstate operations, but with exact preservation of the energy. We say that $\rho_f$ can be reached by \emph{commuting macrostate operations} from the macrostate $(e , H)$, similarly to Definition\ \ref{dfn:macro}, but imposing, instead of mean energy conservation as in Eq.\ \eqref{eq:energyconservation}, the condition \eqref{eq:exactcommutation}. One can define, analogously, \emph{commuting microstate operations} by imposing similarly Eq.\ \eqref{eq:exactcommutation} and a notion of operational equivalence $\overset{\text{c}}{\sim}_\beta$ analogous to \eqref{eq:equiv}.

In Appendix \ref{app:linear}, we show that for every $\beta$ and non-trivial $H$, there exists at least one initial value $e$, such that
\begin{equation}\label{eq:breakdown}
(e , H)\: \overset{\text{c}}{\nsim}_\beta \:\gamma_{e}(H).
\end{equation}
We believe the proof of this result to be interesting in its own right, because in it we show that the maps produced by commuting macrostate operations admit a simple linear characterization, the details of which are discussed in the appendix. Again, an analogous breakdown of the equivalence as given by \eqref{eq:breakdown} exists for several commuting observables.

With respect to the justification of the use of maximum-entropy ensembles, this result implies that one \emph{cannot} justify, in general, assigning a maximum-entropy state to a system under partial information by means of considering the possible thermodynamic transitions in a setting of exact energy conservation. This, we submit, again confirms current practice, because canonical ensembles are rarely used in situations where full control is had over the microdynamics of a system. Moreover, note that from an operational point of view the setting of commuting macrostate operations appears unnatural, because in it one assumes that an experimenter has no access to the microstate information at the level of the systems,  while having full microstate level control over the operations that she implements.

\subsection{The macroscopic limit}

In the light of the inequivalence of macrostates and their respective ensembles for the case of exact commutation, it is interesting to quantify by how much one has to violate \eqref{eq:exactcommutation} in order to recover equivalence. For this, let us introduce the random variable $X$ which quantifies the energy change of $SE$ during a macrostate operation. This energy change is captured by a probability distribution $P$.
Theorem~\ref{thm:main} implies the equivalence between the macrostate $(e,H)$ and its corresponding ensemble with macrostate operations. These preserve the mean energy of the compound, hence with vanishing value of the first moment of $P$, although higher moments could well be different from zero. On the other hand, in the case of commuting macrostate operations, all the higher moments of $P$ would indeed vanish due to condition \eqref{eq:exactcommutation}. Hence, the deviation from zero of the higher moments of $P$ seems a sensible quantifier of the violation of \eqref{eq:exactcommutation}. 

We will now discuss the behaviour of these higher moments for large, non-interacting and independent systems, capturing the classical limit of macroscopic systems. To do so, consider a system $S$ described by $N$ non-interacting subsystems. We will consider macrostate operations between a macrostate $(e,H)$ and a final state $\rho_f$ and impose that the final and initial states are large and uncorrelated. That is, instead of being any microstate in $[e]_H$, the initial microstate takes the form $\sigma=\bigotimes_i^{N} \sigma^i$. We also assume that the final state takes a similar form $\rho_f=\bigotimes_i^N \rho_f^i$. Using standard arguments of central limit theorems one can show that, in the limit of large $N$ and for bounded Hamiltonians, $P(X)$ for the transition $(e,H) \overset{\text{mac}}{\rightarrow} \rho_f$ converges in distribution to a normal distribution with variance scaling as $\sqrt{N}$. Hence, the higher moments of $P(X)$ \emph{per particle} vanish (see Appendix \ref{app:macro}).
This is an argument in favour of the assignment of the ensemble to macrostates, for large weakly-correlated systems, as long as one tolerates violations of \eqref{eq:exactcommutation} -- as measured by the higher moments -- that are negligible in comparison with the typical energy scales involved in the thermodynamic operation.

\section{Conclusion} 
\label{sec:conclusion}
In this work, we have introduced a fresh way of justifying the very common use of maximum-entropy ensembles as a representation of the state of systems. We take a strictly operational stance to the subject, in which an experimenter
has only partial information about the microstate of the system and all
operations have to be compatible with such partial information. The vantage point for our argument concerns the possible thermodynamic transitions that systems can possibly undergo. This approach has the key advantages
that it (a) naturally fits with many operational tasks in thermodynamics and its laws and (b) does not require underlying typicality arguments, and hence avoids some of their conceptual issues.
We have also shown how our results can be used to derive features of phenomenological thermodynamics, such as the Gibbs entropy, free energy as state functions and the Clausius inequality, which determines whether a state transition on macrostates is possible without investing non-equilibrium resources. We are thus able to derive fundamental thermodynamic results without any assumption about typicality or information measures. Finally, our results generalise to the setting of several commuting observables. As such, the results here are likely to be of interest for thermodynamics in generalised settings or even outside the context of thermodynamics.

\section{Acknowledgements} 
\label{sec:acknowledgements}
We thank H.\ Tasaki for comments.
This work has been supported by the ERC (TAQ), the DFG (GA 2184/2-1, CRC 183, B02), the
Studienstiftung des Deutschen Volkes, the EU (AQuS),
and the COST action MP1209 on quantum thermodynamics.

\bibliographystyle{unsrt}
\bibliography{partial_inf.bib}

\begin{thebibliography}{10}

\bibitem{Gibbs1902Elementary}
J.~W. Gibbs.
\newblock {\em Elementary principles in statistical mechanics}.
\newblock Chaeles Sribner's Sons, New York, 1902.

\bibitem{Uffink}
J.~Uffink.
\newblock Compendium of the foundations of classical statistical physics.
\newblock 2015.

\bibitem{Haar}
D.~T. Haar.
\newblock Foundations of statistical mechanics.
\newblock {\em Rev. Mod. Phys.}, 27:289 -- 338, 1955.

\bibitem{Goldstein2006Canonical}
S.~Goldstein, J.~L. Lebowitz, R.~Tumulka, and N.~Zangh\`\i.
\newblock Canonical typicality.
\newblock {\em Phys. Rev. Lett.}, 96:50403, 2006.

\bibitem{Tasaki}
S.~Goldstein, T.~Hara, and H.~Tasaki.
\newblock The second law of thermodynamics for pure quantum states.
\newblock {\em arXiv:1303.6393}, 2013.

\bibitem{Popescu2006Entanglement}
S.~Popescu, A.~J. Short, and A.~Winter.
\newblock Entanglement and the foundations of statistical mechanics.
\newblock {\em Nature Phys.}, 2:754--758, 2006.

\bibitem{XuerebReview}
J.~Millen and A.~Xuereb.
\newblock Perspective on quantum thermodynamics.
\newblock {\em New J. Phys.}, 18:011002, 2016.

\bibitem{ThermoReview}
J.~Goold, M.~Huber, A.~Riera, L.~del Rio, and P.~Skrzypczyk.
\newblock The role of quantum information in thermodynamics.
\newblock {\em J. Phys. A}, 49:143001, 2016.

\bibitem{1408.5148}
J.~Eisert, M.~Friesdorf, and C.~Gogolin.
\newblock Quantum many-body systems out of equilibrium.
\newblock {\em Nature Phys}, 11:124--130, 2015.

\bibitem{ngupta_Silva_Vengalattore_2011}
A.~Polkovnikov, K.~Sengupta, A.~Silva, and M.~Vengalattore.
\newblock Nonequilibrium dynamics of closed interacting quantum systems.
\newblock {\em Rev.\ Mod.\ Phys.}, 83:863--883, 2011.

\bibitem{christian_review}
C.~Gogolin and J.~Eisert.
\newblock Equilibration, thermalisation, and the emergence of statistical
  mechanics in closed quantum systems.
\newblock {\em Rep. Prog. Phys.}, 79:56001, 2016.

\bibitem{Rio2015}
Lidia del Rio, Lea Kraemer, and Renato Renner.
\newblock Resource theories of knowledge.
\newblock {\em arXiv preprint arXiv:1511.08818}, 2015.

\bibitem{Kramer2016}
L.~Kr{\"{a}}mer and L.~{Del Rio}.
\newblock {Currencies in resource theories}.
\newblock {\em arXiv:1605.01064}, 2016.

\bibitem{Horodecki2011Fundamental}
M.~Horodecki and J.~Oppenheim.
\newblock Fundamental limitations for quantum and nanoscale thermodynamics.
\newblock {\em Nat. Commun.}, 4:2059, 2013.

\bibitem{Brandao2015Second}
F.~G. S.~L. Brand\~{a}o, M.~Horodecki, N.~H.~Y. Ng, J.~Oppenheim, and
  S.~Wehner.
\newblock The second laws of quantum thermodynamics.
\newblock {\em PNAS}, 112:3275.

\bibitem{Brandao2013Resource}
F.~G. S.~L. Brand\~{a}o, M.~Horodecki, J.~Oppenheim, J.~M. Renes, and R.~W.
  Spekkens.
\newblock Resource theory of quantum states out of thermal equilibrium.
\newblock {\em Phys. Rev. Lett.}, 111:250404, 2013.

\bibitem{Skrzypczyk2014work}
P.~Skrzypczyk, A.~J. Short, and S.~Popescu.
\newblock Work extraction and thermodynamics for individual quantum systems.
\newblock {\em Nature Commun.}, 5:4185, 2016.

\bibitem{Guryanova2016GGE}
Y.~Guryanova, S.~Popescu, A.~J. Short, R.~Silva, and P.~Skrzypczyk.
\newblock Thermodynamics of quantum systems with multiple conserved quantities.
\newblock {\em Nature Commun.}, 7:12049, 2016.

\bibitem{Halpern2015Beyond}
N.~Yunger Halpern.
\newblock Beyond heat baths ii: Framework for generalized thermodynamic
  resource theories.
\newblock {\em arXiv:1409.7845}, 2014.

\bibitem{Note1}
This is indeed unproblematic since $e$ and $\beta $ are in one to one
  correspondence, hence we will use $\beta $ or $e$ indistinctively when it is
  clear from the context.

\bibitem{Jaynes1957Information}
E.~Jaynes.
\newblock Information theory and statistical mechanics.
\newblock {\em Phys. Rev.}, 106:620--630, 1957.

\bibitem{Rigol_etal08}
M.~Rigol, V.~Dunjko, and M.~Olshanii.
\newblock Thermalization and its mechanism for generic isolated quantum
  systems.
\newblock {\em Nature}, 452:854--858, 2008.

\bibitem{Llobet2015Work}
M.~Perarnau-Llobet, A.~Riera, R.~Gallego, H.~Wilming, and J.~Eisert.
\newblock {Work and entropy production in generalised Gibbs ensembles}.
\newblock {\em New J. Phys.}, 18:123035, 12 2016.

\bibitem{Halpern2016Microcanonical}
N.~Yunger~Halpern, P.~Faist, J.~Oppenheim, and A.~Winter.
\newblock Microcanonical and resource-theoretic derivations of the thermal
  state of a quantum system with noncommuting charges.
\newblock {\em Nature Comm.}, 7:12051, 7 2016.

\bibitem{Lostaglio2017}
M.~Lostaglio, D.~Jennings, and T.~Rudolph.
\newblock Thermodynamic resource theories, non-commutativity and maximum
  entropy principles.
\newblock {\em New J. Phys.}, 19:043008, 2017.

\bibitem{Sparaciari2016}
C.~Sparaciari, J.~Oppenheim, and T.~Fritz.
\newblock A resource theory for work and heat.
\newblock {\em arXiv:1607.01302}, 2016.

\bibitem{Bera2017}
M.~N. Bera, A.~Riera, M.~Lewenstein, and A.~Winter.
\newblock Thermodynamics as a consequence of information conservation.
\newblock {\em arXiv:1707.01750}, 2017.

\bibitem{Jaynes1957Informationa}
E.~Jaynes.
\newblock Information theory and statistical mechanics. ii.
\newblock {\em Phys. Rev.}, 108:171--190, 1957.

\bibitem{Pusz1978Passive}
W.~Pusz and S.~L. Woronowicz.
\newblock {Passive states and KMS states for general quantum systems}.
\newblock {\em Commun. Math. Phys.}, 58:273--290, 1978.

\bibitem{Janzing2000Thermodynamic}
D.~Janzing, P.~Wocjan, R.~Zeier, R.~Geiss, and Th. Beth.
\newblock {Thermodynamic cost of reliability and low temperatures: Tightening
  Landauer's principle and the second law}.
\newblock {\em Int. J. Th. Phys.}, 39:2717--2753, 2000.

\bibitem{RenesCost}
J.~M. Renes.
\newblock Work cost of thermal operations in quantum and nano thermodynamics.
\newblock {\em Eur. J. Phys. Plus}, 129:153, 2014.

\end{thebibliography}

\appendix

 \section{General maximum entropy ensembles}\label{sec:app:gge}
In this section we generalize the formalism laid out in Section  \ref{sec:definitionsmain} to the case of many conserved quantities. That is, the macrostate and microstate operations and the notion of operational equivalence are generalised to the more general case of a set $\{Q^j\}$ of $n$ commuting observables replacing $H$, and a set $\{v^j\}$ of expectation values for each observable replacing $e$. We introduce the following notation to arrange these sets into vectors
\begin{eqnarray}
\mc Q = ( Q^1, \ldots, Q^n),\\
\v= (v^1,\ldots, v^n),
\end{eqnarray}
so the macrostate of the system is given by $(\v,\mc Q)$. The equivalence class of quantum states compatible with the macrostate is denoted by $[\v]_{\mc Q}$.

We model the environment with an analogous assumption as i) in the main text, but for the case of more conserved quantities. We assume that one can have access to $N$ uncorrelated subsystems described each by a macrostate. The mean value of the conserved quantities is determined by the value of a vector of inverse ``temperatures'' $\b=(\beta_1,\ldots,\beta_n)$ for each conserved quantity. We denote, say for subsystem $E^l$, the conserved quantities and mean values as
\begin{eqnarray}
\mc Q_{E^l} &=& ( Q^1_{E^l}, \ldots, Q^j_{E^l}),\\
\v_{\b}(\mc Q_{E^l})&=& \big(v_{\b}(Q^1_{E^l}),\ldots, v_{\b}(Q^j_{E^l})\big),
\end{eqnarray}
where we are making the slight abuse of notation to identify
\begin{equation}
Q^1_{E^l} \equiv \id_1 \otimes \cdots \otimes Q^1_{E^l} \otimes \cdots \otimes \id_N.
\end{equation}
In this way, we will denote the $j$-th conserved quantity on the whole environment as $Q_{E}^j=\sum_{l=1}^N Q_{E^l}^j$. Note that $Q_{E}^j$ plays a similar role as the Hamiltonian of the environment $H_E$ in the main text, but in this case for a different conserved quantity. Accordingly we also can arrange the conserved quantities of the environment, and the compound $SE$ in a vector as
\begin{eqnarray}
\mc Q_E &=& (Q_{E}^1, \ldots, Q_E^n) ,\\
\mc Q_{SE} &=& (Q^1+Q_{E}^1, \ldots, Q^n +Q_E^n).
\end{eqnarray}

The environment is modeled by any macrostate of the form $\bigotimes_{l=1}^N (\v_{\b}(\mc Q_{E^l}),\mc Q_{E^l})$ where, in analogy to equation \eqref{eq:thermalenergy}, we assign a mean value of the conserved quantities equal to the ``thermal'' value, which in this case corresponds to the value that a maximum-entropy ensemble takes. That is,
\begin{equation}\label{eq:macrostate_environment}
v_{\b}(Q^j_{E^l})=\tr\left( \gamma_{\b}(\mc Q_{E^l} )\:Q^j_{E^l} \right),
\end{equation}
where $\gamma_{\b}$ is the so-called \emph{generalised Gibbs ensemble} defined as
\begin{equation}
\gamma_{\b}(\mc Q_{E^l}):=\frac{e^{-\sum_j \beta_j Q^j_{E^l}}}{ \tr\left(e^{-\sum_j \beta_j Q^j_{E^l}}\right) }.
\end{equation}
We are now in a position to introduce macrostate operations.

\begin{definition}[Macrostate operations with many charges] \label{def:macrogge}We say that $\rho_f$ can be reached by macrostate operations from $(\v,\mc Q)$, which we denote by
\begin{equation}
(\v,\mc Q)  \overset{\b\text{-mac}}{\to} \rho_f,
\end{equation}
 if for any $\epsilon>0$ and $\epsilon'>0$ there exist an environment with observables $\mc Q_E$, and a unitary on $SE$ such that
\begin{equation}\label{eq:transitiongge}
\left\| \tr_E (U (\rho_i \otimes \rho_{E^1} \otimes \cdots \otimes \rho_{E^m}) U^{\dagger})-\rho_f\right\|_1 \leq \epsilon,
\end{equation}
while preserving the global value of all the charges
\begin{equation}\label{eq:conservationgge}
\left|\Tr\left(U (\rho_i \bigotimes_{l=1}^N \rho_{E^l} ) U^{\dagger}\: Q_{SE}^j\right)
-\Tr\left( \rho_i \bigotimes_{l=1}^N \rho_{E^l} \:Q_{SE}^j\right)\right| \leq \epsilon',
\end{equation}
for all $j=1,\ldots,n$. Importantly, both \eqref{eq:transitiongge} and \eqref{eq:conservationgge} have to be fulfilled for all the states of $S$ and $E$ compatible with our partial information, that is,
\begin{equation}\label{eq:all_states_in_macro}
\forall \rho_i \in \nonumber[\v]_{\mc Q}, \: \rho_{E^l} \in [\v_{\b} (\mc Q_{E^l} )]_{\mc Q_{E^l}} \text{ for } l\in[1,\ldots,N].
\end{equation}
\end{definition}
{At this point, it is worth briefly discussing the physical significance of $\mc Q$ and $\mc Q_E$. Our framework and in particular our main result -- \emph{i.e.} the equivalence with the maximum entropy ensemble presented in Theorem \ref{thm:gge} -- apply for any choice of charges for $S$ and the environment $E$, given by $\mc Q$ and $\mc Q_{E}$ respectively, as long as the total mean value of the compound is preserved.  In this sense our results leave open and completely general the choice of conserved quantities. However, one must be cautious by noting that imposing a conservation law of the mean value of $\mc Q+\mc Q_{E}$ is not always well-justified.
For instance, when $\mc Q$ are the Hamiltonian, angular momentum and number of particles, it makes sense to allow for environments where $\mc Q_{E^l}$ are the Hamiltonian, angular momentum and number of particles of $E^l$ respectively. In this scenario, imposing \eqref{eq:conservationgge} is meaningful. On the contrary if we take $\mc Q$ to be the angular momentum and $\mc Q_{E}$ to be, say, the magnetisation, we find that it might be in general unjustified to impose a conservation of $\mc Q +\mc Q_{E}$, since those two quantities are, a priori, unrelated. In summary, our framework takes as a starting point that a conservation law is imposed and builds upon this law. The prior arguments that justify imposing such a conservation law are outside the scope of this paper and must be considered independently.
}

The definition of $\rho  \overset{\b\text{-mic}}{\to} \rho_f$ is completely analogous to the case of the previous section, with the GGE ensemble \eqref{eq:ggedef} playing the role of the canonical ensemble.

\begin{definition}[Microstate operations with many charges]\label{def:microgge} We say that
$\rho_f$ can be reached from $\rho_i$ by microstate operations, which we denote by
\begin{equation}
\rho_i  \overset{\b\text{-mic}}{\to} \rho_f,
\end{equation} if for any $\epsilon>0$ and $\epsilon'>0$ there exist an environment with observables $\mc Q_E$ and a unitary on $SE$ such that
\begin{equation}
\left\|\left( U (\rho_i \otimes \gamma_{\b}(Q_E) ) U^{\dagger}\right)-\rho_f \right\|_1 \leq \epsilon,
\end{equation}
while preserving the overall value of the charges
\begin{equation}
  \left|\Tr\left(U (\rho_i \otimes \gamma_{\b}(Q_E) )U^{\dagger}Q_{SE}^j\right)
  -\Tr\left( \rho_i \otimes \gamma_{\b}(Q_E) \:Q_{SE}^j\right)\right| \leq \epsilon',
\end{equation}
for all $j=1,\ldots,n$.
\end{definition}
We can now formulate the main result for the case of multiple observables:

	 \begin{theorem}[Equivalence with the GGE] \label{thm:gge}
	 Let $\mc Q$ be any set of commuting observables and the environment be such that $\beta^j \neq 0$ for all $j$. The macrostate $(\mathbf{v},\mc Q)$ is operationally equivalent to the corresponding GGE ensemble compatible with the partial information $\v$. That is,
	 	\begin{align}
	 	    (\mathbf{v},\mc Q) &\sim_{\b} \gamma_{\v}(\mc Q), \label{eq:main_gge}
	 	\end{align}
	 	where $\v$ are the inverse Lagrange multipliers that one assigns to $S$ so that $\tr(Q^j\gamma_{\v}(\mc Q))=v^j$ for all $j$.
	 \end{theorem}

\section{Proof of Theorems \ref{thm:main} and \ref{thm:gge}} 
\label{app:main}
In this section, we will prove Theorem~\ref{thm:gge}, which implies Theorem~\ref{thm:main} as a special case. The equivalence relation \eqref{eq:main_gge} requires showing that
\begin{equation} \label{eq:equivalencewelookfor}
  (\v,\mc Q) \overset{\b\text{-mac}}{\to} \rho_f \Leftrightarrow \gamma_{\v}(\mc Q) \overset{\b\text{-mic}}{\to} \rho_f.
\end{equation}
The direction ``$\Rightarrow$'' is trivial. Note that the l.h.s. implies that the transition is possible for all initial states compatible with $(\v,\mc Q)$. In particular, $\gamma_{\v}(\mc Q)$ is one of these states compatible with $(\v,\mc Q)$ and hence the r.h.s.\ condition follows.

Before embarking on the proof of the direction ``$\Leftarrow$'', we will provide an overview of the different steps involved:
\begin{enumerate}
\item We show that macrostate operations allow us to consider without loss of generality probabilistic mixtures of unitary operations as well.

\item We show that using probabilistic mixtures of unitaries, we can reduce the problem to only considering microstates which are diagonal in the basis of the conserved quantities.

\item Using the previous results we show that we can "distill", from the environment described by partial information, systems for which we are certain that they are in the microstates given by the GGE to arbitrary accuracy and with arbitrarily little change of the charges. This shows that we can effectively describe the environment by GGE microstates directly.

\item We show that once we have an environment directly described by GGE microstates, we can always bring the system to the GGE microstate corresponding to its macrostate. That is, we show that it is possible to implement the transition
\begin{equation}\label{eq:transitiontothermal}
(\v,\mc Q) \overset{\b\text{-mac}}{\to}\gamma_{\v}(\mc Q).
\end{equation}
\end{enumerate}
Finally, after we have replaced the state on the system with the GGE by a macrostate operation, we can apply the microstate operation that maps the GGE to the desired final state (r.h.s.\ of \eqref{eq:equivalencewelookfor} which is the premise of the proof). That is, we compose macrostate operations and microstate operations in the following way:
\begin{equation}\label{eq:compositionrule}
  (\v,\mc Q) \overset{\b\text{-mac}}{\to} \rho \:\: \land \:\: \rho \overset{\b\text{-mic}}{\to} \sigma \Rightarrow (\v,\mc Q) \overset{\b\text{-mac}}{\to} \sigma.
\end{equation}
By taking $\rho= \gamma_{\v}(\mc Q)$ and $\sigma=\rho_f$ and using \eqref{eq:transitiontothermal} we obtain the direction ``$\Leftarrow$'' of \eqref{eq:equivalencewelookfor} which concludes the proof.
We will now give detailed derivations of steps 1.-4. separately.

\subsection{Mixtures of unitaries}\label{sec:mixture_unitaries}
We will now show that instead of considering unitary operations for macrostate operations, for finite temperature environments, we can also use probabilistic mixtures of unitaries. The basic idea is to use systems from the environment, described by the macrostate $\bigotimes_{l=1}^N (\v_{\b}(\mc Q_{E^l}),\mc Q_{E^l})$, as a source of randomness.

Suppose we want to act with a mixture of unitaries on some $m$ systems at hand (which might include other systems from the environment). To do that, we first take two additional systems out of the environment. We choose these subsystems to be qubits with $\mc{Q}_{E^l}=(H,\id,\ldots ,\id)$. That is, we only consider the energy as a conserved quantity. Let us re-scale their Hamiltonian so that we can write it as $H=0 \ketbra{0}{0} +\Delta\ketbra{1}{1}$. As the macrostates have energy $e_{\beta}(H)$, this determines that $[e_{\beta}(H)]_H$ is formed by states with $\tr(\rho \ketbra{0}{0}):=p_0$, $\tr(\rho \ketbra{1}{1}):=p_1$ with $p_1=e_{\beta}(H)/\Delta$. Let us choose $\Delta$ so that $p_0=1/\sqrt{2}$.

We now apply to the $m$ subsystems the unitary
\begin{equation}
 U=\ketbra{0,0}{0,0} \otimes U_{\text{rest}} + (\ketbra{0,1}{0,1}+\ketbra{1,0}{1,0}+\ketbra{1,0}{1,0})\otimes U'_{\text{rest}}
\end{equation}
where ``rest'' refers to the $m$ subsystems upon which we want to apply the mixture of unitaries. One obtains that the effective map on the $m$ systems is
\begin{eqnarray}
\rho&\mapsto& \mc{M}(\rho) =\tr_{\text{2-qub}}( U \rho U^{\dagger})\\
\nonumber &=&p_{0,0} U_{\text{rest}} \rho U^{\dagger}_{\text{rest}} + (p_{0,1}+p_{1,0}+p_{1,1}) U'_{\text{rest}}
\rho U'^{\dagger}_{\text{rest}}\\
\nonumber &=&(p_0)^2 U_{\text{rest}} \rho U^{\dagger}_{\text{rest}} + (1-(p_0)^2) U'_{\text{rest}}
\rho U'^{\dagger}_{\text{rest}}\\
\nonumber &=&\frac{1}{2} U_{\text{rest}} \rho U^{\dagger}_{\text{rest}} +\frac{1}{2} U'_{\text{rest}}
\rho U'^{\dagger}_{\text{rest}}.
\end{eqnarray}
Repeating this process with as many pairs of qubits as required, we can apply apply any mixture of unitaries that we need. Hence, we can assume without loss of generality that in order to perform a macrostate operation as given by Definition\ \ref{def:macrogge}, it suffices to find, instead of a single unitary $U$ on the $SE$ compound, a mixture of unitaries that performs the desired transition, which we denote as
\begin{equation}\label{eq:unitarymap}
\rho\mapsto \mc U ( \rho) =\sum_\lambda p_\lambda U_\lambda \rho U_\lambda^{\dagger},
\end{equation}
with each of $U_\lambda$ preserving the mean value of the conserved quantities.

\subsection{Reducing the problem to diagonal microstates}\label{sec:reducing_diagonal}
We now show that by being able to implement mixtures of energy-preserving unitaries, we can reduce the problem to one
in which all microstates are diagonal in the eigenbasis of all the conserved quantities. To do that, define for every operator $Q^j$ the mixture of unitaries
\begin{align}\label{eq:decoherence_map}
\rho\mapsto \mc D_{Q^j}(\rho) := \lim_{T\rightarrow \infty} \frac{1}{T}\int_{0}^T \e^{\mathrm{i}Q^j t}\rho\e^{-\mathrm{i}Q^j t}\, \mathrm{d}t.
\end{align}
This mixture of unitaries dephases every state in the eigenbasis of $Q_j$. Since all the $Q_j$ commute, we can sequentially apply these maps to map any state $\rho\in [\v]_{\mc Q}$ to a state that commutes with all $Q^j$. In the following, we will denote this set of microstates that are diagonal in the eigenbasis of all the $Q^j$ and correspond to the macrostate $(\v,\mc Q)$ by $[\v]_{\mc Q}^\mathrm{diag}$. The fact that we can dephase all states without changing the mean values $v^j$ implies that condition Eq.~\eqref{eq:all_states_in_macro} of Definition~\ref{def:macrogge} can be relaxed to diagonal states, i.e.,
\begin{equation}
\forall \rho_i \in \nonumber[\v]_{\mc Q}^\mathrm{diag}, \: \rho_{E^l} \in [\v_{\b} (\mc Q_{E^l} )]_{\mc Q_{E^l}}^\mathrm{diag} \text{ for } l\in[1,\ldots,N].
\end{equation}
This allows us to restrict to diagonal states in the last two steps (3. and 4.).

\subsection{From the macrostate environment to the  maximum entropy environment}
\label{sec:macrotomicro}
The macrostate operations and the microstate operations employ different models of the environment. As discussed in the main text (see Section  \ref{sec:app:gge} for the generalisation for many conserved quantities),
the environment for macrostate operations is given by macrostates of the form
\begin{equation}\label{eq:macrobath}
\bigotimes_{l=1}^{N} (\v_{\b}(\mc Q_{E^l}),\mc Q_{E^l}).
\end{equation}
On the other hand, for microstate operations one assumes that the environment is given by maximum entropy ensembles of the form
\begin{equation}
\bigotimes_{l=1}^{N'} \gamma_{\b}(\mc{Q}_{E^l}). \label{eq:microbath}
\end{equation}
We will now show that any environment of the form \eqref{eq:microbath} can always be ``distilled'' from an environment of the form \eqref{eq:macrobath}. That is, for any $N'$ one can always find a sufficiently large $N$ so that a system of the form \eqref{eq:microbath} is obtained.

Due to the results of Section  \ref{sec:mixture_unitaries} and \ref{sec:reducing_diagonal} we can, without loss of generality, model the macrostate operations that achieve this distillation by mixtures of unitaries that act on diagonal states of the bath, requiring only that they preserve the total expectation values of all the observables. For simplicity, we will take $N'=1$, since an extension to larger values of $N'$ can be done by simply repeating the process over $N'$ copies of \eqref{eq:macrobath}.

For purely technical reasons, we will for now consider the special case where the eigenvalues of all the conserved quantities $Q_{E^l}^j$ have rational eigenvalues. Since any operator can be approximated to arbitrary accuracy by one with rational eigenvalues, this is not a severe restriction.

Consider a larger number $N$ of identical environment systems in the same macrostate $(\v_\beta(\mc Q_{E^l}),\mc Q_E^{l})$, where $\mc Q_{E^l}=\mc Q_{E^{l'}}$ for all $l,l'=1,\ldots,N$. We will apply a unitary map $\mc{U}$ of the form \eqref{eq:unitarymap} and find that the reduced state on every subsystem is given by $\gamma_{\b(\mc Q_{E^l})}$ to arbitrary accuracy as $N\rightarrow \infty$.

We first have to set up some notation. A basis-state on one of the subsystems can be labelled by the eigenvalues $q^j_\alpha$ of the $n$ conserved quantities $Q^j_{E^l}$, where $\alpha=1,\ldots,d_{E^l}(j)$ and $j=1,\ldots,n$. Here, $d_{E^l}(j)$ is the number of \emph{distinct} eigenvalues of $Q^j_{E^l}$. Simplifying the notation, the basis states on system $E^l$ can thus be labeled by $d$ vectors $\a^x = (\alpha^x_1,\ldots,\alpha^x_n)$ corresponding to the choice of eigenvalues $q^j_{\alpha^x_j}$. A basis-state for the $N$ systems is then given by choosing one vector $\a^x$ for each subsystem and is denoted by $\a^{\x} = (\a^{x_1},\ldots,\a^{x_N})$. We will label the joint-eigenspaces of the $Q^j_E$ on the $N$ systems by $\Pi_\xi$ and identify also $\Pi_\xi$ with the projector onto that eigenspace. Given an eigenspace $\Pi_\xi$, we finally denote the corresponding eigenvalue of the total charge $Q^j_E$ as $q^j_{E,\xi}$.

After setting up the notation, we will now start with the actual proof. The operation that we consider is very simple: We simply apply a completely random unitary in each of the subspaces $\Pi_\xi$. This operation clearly commutes with the total charges, hence it also preserves its average value. If we denote the total probability of subspace $\Pi_\xi$ by $p_\xi$, it leaves the whole distribution $p_\xi$ invariant, while leaving each of the subspaces in the maximally mixed state $\Omega_\xi$. Since each of the subspaces is permutation invariant, we find that the state of every system is finally described by the same density matrix
\begin{align}
\rho'_{E^l} = \sum_{\xi} p_\xi \Tr_{\overline{l}}(\Omega_\xi).
\end{align}
Since the \emph{initial} state $\otimes_l \rho_{E^l}$ is uncorrelated, the total weight of joint eigenspaces $\Pi_\xi$ for which any of the eigenvalues $q^j_{E,\xi}$ deviates by more than $O(\sqrt{N})$ from $N v^j_\beta$ is exponentially small (by Hoeffding's inequality). We will collect the remaining subspaces in a set $\mc M$. We thus have
\begin{align}
\rho'_{E^l} = \sum_{\xi\in\mc M}p_\xi \Tr_{\overline{l}}(\Omega_\xi) + \epsilon_N\sigma,
\end{align}
where $\sigma$ is some density-matrix and $\epsilon_N$ goes to zero exponentially with $N$. Note that for all $\xi \in \mc M$ the corresponding eigenvalues fulfill
\begin{align}
|q^j_{E,\xi}/N  - \v_{\b}^j| \leq \delta^j_{N},\quad \delta^j_N\overset{N\rightarrow\infty}{\longrightarrow}0.
\end{align}

We will now show that, as $N\rightarrow \infty$, the reduced state on any single subsystem of each of the maximally mixed states $\Omega_\xi$, with $\xi\in\mc M$, approaches the GGE. To see this pick any such subspace $\Pi_\xi$. The fact that the eigenvalues $q^j_\alpha$ are all rational, together with the fact that $\xi\in\mc M$ implies that the dimension of any such subspace becomes arbitrarily large with increasing $N$.

Now consider the basis vectors $\a^\x=(\a^{x_1},\ldots,\a^{x_N})$ in $\Pi_\xi$. We will associate to each such basis vector a \emph{type}
\begin{align}
T( \a^{\x} ) = \left(\frac{k_1}{N},\ldots,\frac{k_d}{N} \right),
\end{align}
where $k_x$ is the number of subsystems in state $\a^x$. In other words, they fulfill $\sum_x k_x=N$ and
\begin{align}
  \sum_{x=1}^{d_{E^l}(j)} k_x q^j_{\alpha^x_j} = q^j_{E,\xi}.
\end{align}

The number of basis vectors corresponding to the same type $T$ is given by
\begin{equation}
	\#T = \frac{N!}{\prod_{x=1}^d k_d}.
\end{equation}
It can be bounded using Stirling's approximation as
\begin{align}
\sqrt{2\pi} \poly{N}\e^{N S(T)} \leq \#T  \leq \e\, \poly{N}\e^{N S(T)},\nonumber
\end{align}
where $S(T)=S(k_1/N,\ldots,k_d/N)$ is the Shannon-entropy of a type. Note that the total dimension of one eigenspace $\Pi_\xi$ is simply given by
\begin{align}
d(\Pi_\xi) = \sum_{T\in \Pi_\xi} \# T.
\end{align}
A type has the property that $T_x=k_x/N \geq 0$ and $\sum_{x=1}^d k_x/N = 1$. It can hence be interpreted as a probability distribution. If the total system is in the state $\Omega_\xi$, we obtain from permutation invariance that the probability to find the $l$-th subsystem in state $\a^x$ is given by
\begin{align}
  p^\xi_{E^l}(\a^x) = \frac{\sum_{T\in\Pi_\xi} T_x \# T}{\sum_{T\in \Pi_\xi} \# T}.
\end{align}

We will now show that all types that differ from the GGE-distribution by more than $\delta$ (in some norm on $\RR^{d-(n+1)}$) have a relative weight that vanishes as $N\rightarrow \infty$. In other words, as we increase the system size, the probability distribution  $p^\xi_{E^l}(\a^x)$ converges to that of a GGE with $\v^j = q^j_{E,\xi}/N$. Let us denote the probability distribution corresponding to the GGE in subspace $\xi$ by $\gamma_\xi$. Since the Shannon entropy is concave and has a unique maximum among all probability distributions compatible with the expectation values of the conserved quantities $Q^j_{E^l}$ corresponding to the subspace $\xi$, we can bound the entropy of any type that differs by more than $\delta$ from $\gamma_\xi$ as
\begin{align}
S(\gamma_\xi) - K'\delta^2 \leq S(T) \leq S(\gamma_\xi) - K\delta^2,
\end{align}
where the constants $K$ and $K'$ do not depend on $N$.

We thus see that the weight of the type is
\begin{align}
\sqrt{2\pi} \poly{N} \e^{N S(\gamma_\xi)-N K'\delta^2}&\leq \# T\nonumber\\& \leq \e\, \poly{N}\e^{N S(\gamma_\xi)-N K\delta^2 }.\nonumber
\end{align}
Hence, the weight of the types is distributed according to a Gaussian-distribution on a subset of $\RR^{d-(n+1)}$ with variance $\sigma^2$ of order $1/N$. For large $N$, it is thus very sharply peaked around the Gibbs-distribution and we can choose $\delta$ to go to $0$ as $N\rightarrow \infty$  while at the same time most of the weight of the distribution is carried by distribution within $\delta$ away from the GGE distribution. Choose, for example, $\delta = N^{1/4} \sigma$, so that
\begin{align}
\lim_{N\rightarrow \infty} N^{1/4}\sigma = \lim_{N\rightarrow \infty} N^{1/4-1/2} = \lim_{n\rightarrow \infty} N^{-1/4} = 0.
\end{align}
More formally, we can upper bound the total weight of types more than $\delta$ away from the GGE distribution by
\begin{align}
  \sum_{\substack{T\in \Pi_\xi, \\ \norm{T-\gamma_\xi}_1\geq \delta}} \# T \leq \mc T_\xi\, e\, \poly{N}\e^{N S(\gamma_\xi)-N K\delta^2 },
\end{align}
where $\mc T_\xi$ is the total number of different types appearing in subspace $\Pi_\xi$. Similarly, for any $q<1$ we can lower bound the total weight of types closer than $q\delta$ to the GGE distribution by
\begin{align}
\sum_{\substack{T\in \Pi_\xi, \\  \norm{T-\gamma_\xi}_1\leq q \delta}} \# T \geq \poly{q\delta} \mc T_\xi\, \sqrt{2\pi}\, \poly{n}\e^{N S(\gamma_\xi)-N K'q^2\delta^2 }.\nonumber
\end{align}
The relative volume of the two is then given by (using $\delta = N^{-1/4}$)
\begin{widetext}
\begin{align}
\frac{\e\, \poly{N} \e^{N S(\gamma_\xi) - NK\delta^2 }}{\sqrt{2\pi}\poly{q\delta} \poly{N} \e^{N S(\gamma_\xi) - NK'\delta^2 q^2}} &= \frac{\e\, \poly{N} \e^{N S(\gamma_\xi) - \sqrt{N}K }}{\sqrt{2\pi}\poly{q N^{-1/4}} \poly{N} \e^{N S(\gamma_\xi) - \sqrt{N}K' q^2}}\\& \leq K'' \poly{N} \e^{-\sqrt{N}(K-K'q^2)} \rightarrow 0,\nonumber
\end{align}
\end{widetext}
for $q< \sqrt{K/K'}$.
As $N\rightarrow \infty$, we therefore find that
\begin{align}
\lim_{N\rightarrow \infty} \tr_{\overline{l}}\left(\Omega_\xi\right) &= \lim_{n\rightarrow \infty} \sum_{x} p^\xi_{E^l}(\a^x) \ketbra{\a^x}{\a^x} \nonumber\\
&= \lim_{N\rightarrow \infty} \gamma_{{\b}_\xi}(\mc Q_{E^l})\nonumber\\
&= \gamma_{\b}(\mc Q_{E^l}),
\end{align}
where ${\b}_{\xi}$ is the vector of "inverse temperatures" corresponding to the subspace $\Pi_\xi$ and in the last line we have used that $\lim_N q^j_{E,\xi}/N = {\v}_{\b}^j$ for all $\xi \in \mc M$. Since this holds for all subspaces in $\mc M$, we finally obtain the desired result that
\begin{align}
\rho'_{E^l} = \sum_{\xi\in\mc M} p_\xi \tr_{\overline{l}}(\Omega_\xi)   + \epsilon_N \sigma \overset{N\rightarrow\infty}{\longrightarrow} \gamma_{\b}(\mc Q_{E^l}).
\end{align}
Concluding, we have shown that by taking many copies of the macrostate $(\v_{\b},\mc Q)$ and applying an exactly energy-conserving operation, we can prepare the microstate $\gamma_{\b}(\mc Q)$. Repeating this process many times, we can then also prepare any environment of the form
\begin{align}
\bigotimes_l \gamma_{\b}(\mc Q_{E^l}).
\end{align}

\subsection{Bringing the system to the maximum entropy state using the maximum entropy environment}
In the last section we have proven that, from the model of the environment given by \eqref{eq:macrobath} for the definition of macrostate operations, one can distill a microstate environment of the form \eqref{eq:microbath}. We will now use such an environment to bring the system to the maximum entropy state. That is, to perform the transition \eqref{eq:transitiontothermal}. The idea to do that is very simple: We choose the right conserved quantities $Q_E$ on the environment and then simply swap the system state with the environment.

Suppose that the system is in macrostate $(\v, \mc Q)$ with conserved quantities $Q^j$ and let the corresponding inverse temperatures given by $\gamma_{\v}(\mc Q)$ be given by $\beta_j(\v)$.
Now choose the following conserved quantities on the environment,
\begin{align}
  Q^j_E = \frac{\beta_j({\v})}{\beta_j} Q^j.
\end{align}
Of course, this is possible only if $\beta_j \neq 0$ for all $j$. Then the two density matrices of the GGEs coincide, $\gamma_{\b}(\mc Q_E) = \gamma_{{\b}({\v})}(\mc Q)$, and hence the total charge is conserved on average as the two states are swapped (it is not conserved exactly, since the microstate on the system can be any microstate in $[\v]_{\mc Q}$). As mentioned in the previous section, the above reasoning strictly speaking only applies if the eigenvalues of $Q^j_E$ are rational. However, we can always approximate $Q^j_E$ by an operator with rational eigenvalues to arbitrary precision. In this case, the average charge conservation is fulfilled with arbitrary precision as well.

\section{Non-Gibbsian average energies trivialize thermodynamics}
\label{sec:trivializing}
In this section, we will show that the assignment of macrostates to the environment as in Eq.\ \eqref{eq:thermalenergy} is the only one that does not lead to i) arbitrary work extraction from the environment and ii) trivial macrostate operations, in the sense that any transition is possible. For this, we will analyse the consequences of having an assignment of energies given by $\f(H)$ different from $e_{\beta}(H)$ as given by \eqref{eq:thermalenergy}. For simplicity we will discuss it for the case of the energy as a single conserved quantity, since the argument is fully analogous for the case of other conserved quantities.

Let us first show i). The function $f$ can, without loss of generality, be always expressed as $f(H)=e_{\beta(H)}(H)$, where now $\beta(H)$ is not a fixed value but a function of the Hamiltonian. For the situation to not be equivalent to some fixed inverse temperature, at least two Hamiltonians must have different temperatures, i.e., there exist Hamiltonians $H_1\neq H_2$ such that $\beta(H_1) \neq \beta(H_2)$. For simplicity let us write $\beta_j=\beta(H_j)$ in the following. Given any value of $\beta_j$ we can repeat the argument of Section  \ref{sec:macrotomicro} and distill, from a large number of macrostates of the environment, one canonical ensemble at temperature $\beta_j$. That is, from an environment of the form
\begin{align}
\bigotimes_{j=1}^{N_1} (e_{\beta_1}(H_1),H_1)\bigotimes_{j=1}^{N_2} (e_{\beta_2}(H_2),H_2)
\end{align}
one can obtain systems in the microstate
\begin{align}\label{eq:twobaths}
\gamma_{\beta_1}(H_1)^{\otimes N'_1}\otimes \gamma_{\beta_2}(H_2)^{\otimes N'_2}
\end{align}
with $N'_1$ and $N'_2$ arbitrarily large for sufficiently large $N_1$ and $N_2$. Once we possess two systems in the canonical ensemble at different inverse temperatures $\beta_1$ and $\beta_2$, one can trivially extract work. That is, one could reduce the mean energy of \eqref{eq:twobaths} and accumulate it in a work storage device. This is true since for some value for $N'_1$ and $N'_2$ \eqref{eq:twobaths} will cease to be a passive state \cite{Pusz1978Passive}.

The previous considerations imply trivially ii). Once we have established that the environment could be used to extract an arbitrary amount of work --mean energy--, one can invest this energy in creating an arbitrary state \cite{Skrzypczyk2014work}. Hence one finds that if $f(H)$ is not the thermal energy, then
\begin{align}
 (e,H) \overset{\text{$\beta$-mac}}{\rightarrow} \rho.
\end{align}
is possible for any $\rho$.

Altogether, we conclude that imposing that i) or ii) are impossible implies that $f(H)=e_{\beta}(H)$ for a fixed $\beta$. In other words, there only exist specific families of functions, one for each value of $\beta$, that do not lead to trivial macrostate operations or work extraction from the environment. In this way the assignment of a parameter $\beta$ to the environment follows from those basic principles. Importantly, note that the parameter $\beta$ is in principle not related to any prior assignment of a temperature to the environment. For the sake of simplicity, we refer to $\beta$ as the inverse temperature, but the interpretation of $\beta$ as related to a prior value of $T$ as $\beta=(k_B T)^{-1}$ is not necessary to derive Theorem~\ref{thm:main} or any of the results in this work.
In summary, we conclude that the only thermodynamically consistent way to assign average energies to environment systems is by assigning the energies corresponding to a thermal Gibbs state for some parameter $\beta$ playing the role of an inverse temperature.

\section{Breakdown of equivalence under exact energy conservation}\label{app:linear}
In this section, we will prove the inequivalence between macrostates and their corresponding maximum-entropy ensemble when exact energy conservation, iv') in Sec.~\ref{sec:breakdown}, is imposed. In particular, we show that for every $\beta$ and non-trivial $H$, there exists at least one initial value $e$, such that
\begin{equation}\label{eq:notequivalent}
(e , H)\: \overset{\text{c}}{\nsim}_\beta \:\gamma_{e}(H).
\end{equation}

Let us first introduce some notation. We define \emph{commuting macrostate operations}, denoted by
\begin{equation}
(e,H) \overset{\beta\text{-c-mac}}{\rightarrow} \rho_f,
\end{equation}
 similarly to Definition\ \ref{dfn:macro} but replacing condition \eqref{eq:energyconservation} by $[U,H_{SE}]=0$. In a similar fashion, we define \emph{commuting microstate operations}, denoted by
 \begin{equation}
\rho \overset{\beta\text{-c-mic}}{\rightarrow} \rho_f,
\end{equation}
similarly to Definition\ \ref{dfn:micro} but replacing condition \eqref{eq:energyconservation} by $[U,H_{SE}]=0$. Commuting microstate operations are in the literature discussed as ``thermal operations''
\cite{Brandao2015Second,Horodecki2011Fundamental}. Proving the inequivalence \eqref{eq:notequivalent} amounts to finding one microstate $\sigma$ so that
\begin{eqnarray}
(e,H) &\overset{\beta\text{-c-mac}}{\not\rightarrow}& \sigma, \\
\gamma_e(H) &\overset{\beta\text{-c-mic}}{\rightarrow}& \sigma.
\end{eqnarray}
The existence of such a state $\sigma$ is implied by the fact that, for any $H$ that admits non-trivial equivalence classes, $|[e]_H|>1$, and any $\beta$, there exists at least one initial energy $e$ such that
	 \begin{equation}\label{eq:max_linear_vs_thermal}
\max_{(e,H) \overset{\beta\text{-c-mac}}{\rightarrow} \rho_f} \mc E (\rho_f) < \max_{\gamma_e(H) \overset{\beta\text{-c-mac}}{\rightarrow} \rho_f} \mc E (\rho_f).
\end{equation}
\eqref{eq:max_linear_vs_thermal} implies the existence of $\sigma$ because, if $\sigma$ did not exist, then the reachable energies under the two types of operations would coincide.
\eqref{eq:max_linear_vs_thermal} itself follows from a result that we present in the next section and in which the reachable energies under macrostate commuting operations are linearly upper bounded, as illustrated in Fig.\  \ref{fig:comp}. We believe that this bound may be of independent interest.

 \begin{figure}
\begin{center}
\includegraphics[scale=0.4]{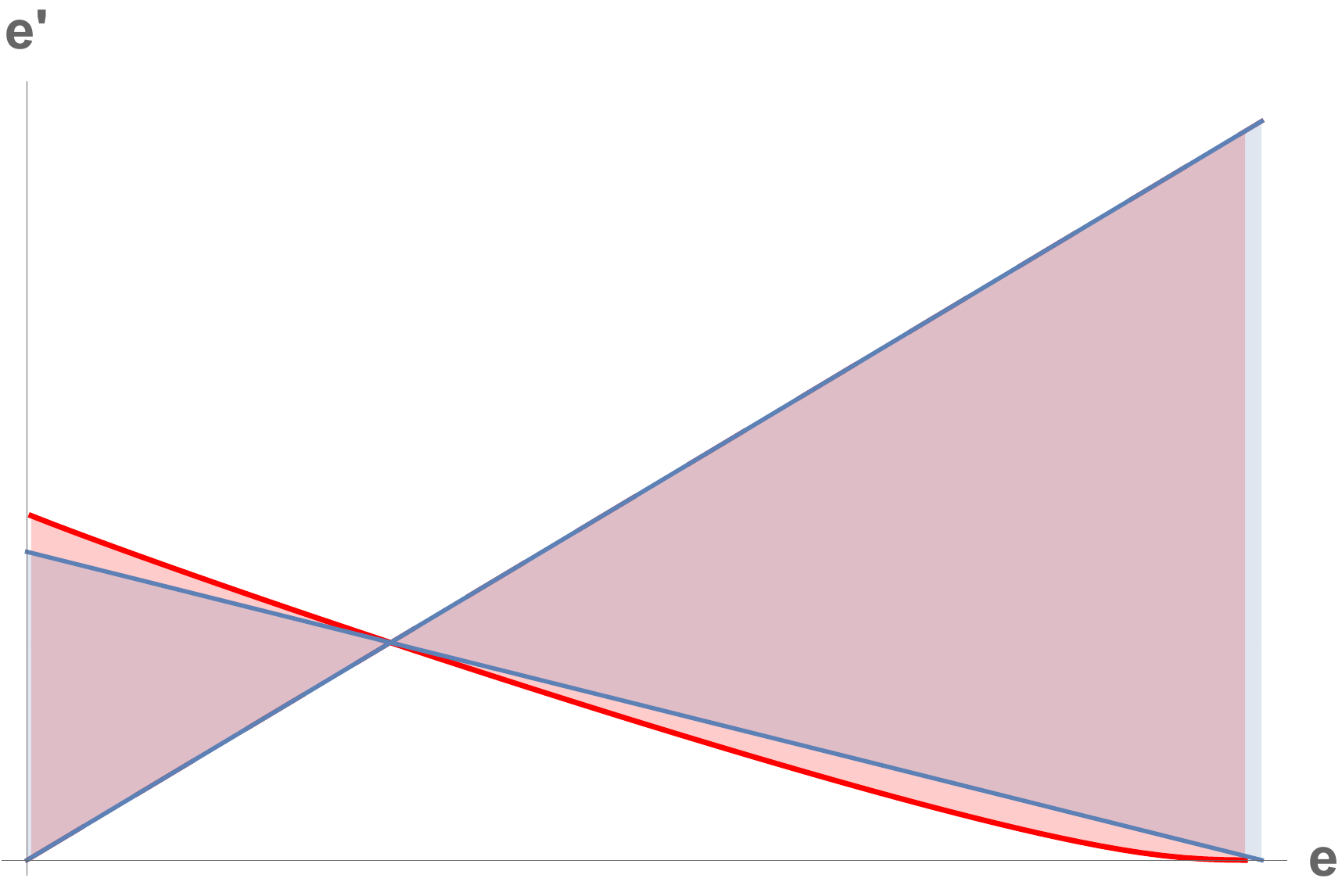}
\caption{The set of reachable final energies $e'$, given some Hamiltonian $H$ and initial energy $e$ (corresponding to a non-trivial equivalence class, $|[e]_H|>1$): The reachable final energies under commuting macrostate operations are \emph{upper} bounded by two lines (blue region) that themselves \emph{lower} bound the set of reachable energies under microstate commuting operations (red region). The results of Refs.\ \cite{Horodecki2011Fundamental} imply that, for any non-trivial $H$ and $\beta$, the red region has a non-linear boundary, which further implies that the blue region is strictly smaller than the red region. This, in turn, immediately gives \eqref{eq:max_linear_vs_thermal} and, hence, yields the breakdown of operational equivalence, \eqref{eq:notequivalent}. In this figure, the intersection point marks the thermal energy $e_\beta(H)$, that is a fixed point of all operations by definition. Note further that the two sets are bounded, in one direction, by the identity. This follows from free energy considerations (see Appendix \ref{sec:linearity}).}
\label{fig:comp}
\end{center}
\end{figure}

\section{Partial characterisation of commuting macrostate transitions}\label{sec:linearity}

 In this section we will provide a method to analyse the allowed transitions under commuting macrostate operations. We cannot in general provide a full answer to which transitions $(e,H) \overset{\beta\text{-c-mac}}{\rightarrow} \rho_f$ are possible. However, we will provide a method to bound the maximum and minimum energies of the states $\rho_f$ achievable from a given macrostate $(e,H)$.

First, we need to consider a set of transitions between macrostates that are closely related to those produced by commuting macrostate operations:
	 \begin{definition}[Macrostate GP-maps] \label{def:gpmaps}
We say that $(e', H)$ can be reached from $(e,H)$ by macrostate GP-maps, which we denote by $(e,H) \overset{\beta\text{-mGP}}{\to} (e', H)$, if for any $\epsilon>0$ there exists a completely positive, trace preserving (CPTP)-map $G$ such that
\begin{enumerate}
	\item $G(\gamma_{\beta}(H)) = \gamma_{\beta}(H),$
	\item $G(\rho) \in [e']^\epsilon_{H}, \quad \forall \rho \in [e]_{H}$ \label{item:cond_on_gpmaps}.
\end{enumerate}
	\end{definition}
Here, $[e']^\epsilon_{H}$ denotes the union of the equivalence classes that differ from $e'$ by at most $\epsilon$.
By definition of the operations, and from results in Ref.\ \cite{Janzing2000Thermodynamic}, the following chain of implications holds: For any $\rho \in [e']_H$,
\begin{align}
    (e,H) \overset{\beta\text{-c-mac}}{\rightarrow} \rho \quad &\Rightarrow (e,H) \overset{\beta\text{-mGP}}{\rightarrow} (e',H), \\
     &\Rightarrow \gamma_e(H) \overset{\beta\text{-c-mic}}{\rightarrow} \gamma_{e'}(H).
\end{align}
This in turn implies that for all $(e,H)$,
	 \begin{align}\label{eq:max_linear_vs_thermal2}
\max_{(e,H) \overset{\beta\text{-c-mac}}{\rightarrow} \rho_f} \mc E (\rho_f) &\leq \max_{(e,H) \overset{\beta\text{-mGP}}{\rightarrow} (e',H)} e' &\leq \max_{\gamma_{e}(H) \overset{\beta\text{-c-mic}}{\rightarrow} \rho_f} \mc E (\rho_f).
\end{align}

From the results of Ref.\  \cite{Horodecki2011Fundamental} it follows that the rightmost term in \eqref{eq:max_linear_vs_thermal2} is a non-linear function of $e$. In contrast, for the middle term, we find the following lemma.

	\begin{lemma}[Reachable energies under macrostate GP-maps] \label{lem:linearity}
	For any non-trivial $H$ and $\beta$, if $|[e]_H|>1$,
	\begin{align}
	    \max_{(e,H) \overset{\beta\text{-mGP}}{\rightarrow} (e',H)} e' =
	 \begin{cases}
	 	e &\text{ if } e \geq e_\beta(H) ,\\
	 	e_\beta(H) + \alpha(e) K_{\beta,H}  &\text{ if } e < e_\beta(H),
	 \end{cases}
	\end{align}
	 where $e\mapsto \alpha(e)$ is a function linear in $e$ and $K_{\beta, H}$ is a constant independent of $e$.
	 Similarly,
	 \begin{align}
	     	 	 \min_{(e,H) \overset{\beta\text{-mGP}}{\rightarrow} (e',H)} e' =
	 \begin{cases}
	 	e_\beta(H) + \alpha(e) K_{\beta, H} &\text{ if } e \geq e_\beta(H) ,\\
	 	e &\text{ if } e < e_\beta(H).
	 \end{cases}
	 \end{align}
	\end{lemma}
This lemma characterizes the set of reachable energies under macrostate GP-maps, and hence upper and lower bounds the possible state transitions under commuting macrostate and microstate operations respectively. As discussed below, the constant $K_{\beta,H}$ can easily be evaluated as a linear program. With respect to \eqref{eq:max_linear_vs_thermal2}, Lemma \ref{lem:linearity} and the results from Ref.\  \cite{Horodecki2011Fundamental} together imply that the second inequality in \eqref{eq:max_linear_vs_thermal2} has to be strict and hence that \eqref{eq:max_linear_vs_thermal} holds.

\subsection{Proof of Lemma \ref{lem:linearity}} 
\label{sub:proof_of_lemma_ref}

Denote the set of macrostate GP-maps for a given initial energy $e$ as $\mc G_e$. First, note that just like in the previous proofs, we need to consider only microstates $\rho \in [e]^{\text{diag}}_{H}$ that are diagonal in the eigenbasis of $H$, because the decoherence map $\mc U_{dec.}$ defined in \eqref{eq:decoherence_map} is clearly a macrostate GP-map (mapping a macrostate to itself). Next, let
\begin{align}
    \N = \{A | \text{diag}(A) = A  \wedge \tr(H^\dagger A) = 0 \wedge \tr(A) = 0 \}
\end{align}
be the space of traceless, diagonal matrices that are orthogonal to $H$, for which $\text{dim}(\N)= d-2$. Further, let $T$ be the matrix that is orthogonal to both $H$ and $\mc N$ and for which $\tr(H)=\tr(T)$. This matrix always exists. Clearly, if $\{ N_i\}_{i=1}^{d-2}$ is some orthogonal basis of $\N$, then $\{H,T,N_1,\dots, N_{d-2}\}$ form a complete basis of the diagonal sector. For this reason, we can expand any diagonal state $\rho$ as
\begin{align}
    \rho = \gamma_{e}(H) + \alpha(e) (H - T) + N(\rho), \label{eq:expand}
\end{align}
where
\begin{equation}
\alpha(e) =\frac{e - e_\beta(H)}{\tr(H^2)},
\end{equation}
$N(\rho) \in \N$. Furthermore, by construction, in this expansion, any two states from the same equivalence class differ only by an element in $\N$. This expansion is useful because it allows us to show the following lemma.

	\begin{lemma}[Characterising initial states in macrostate GP-maps] \label{lem:nton}
	For non-trivial $H$, a CPTP-map satisfies condition \ref{item:cond_on_gpmaps} from Definition\ \ref{def:gpmaps} iff
	$G(\mathcal{N}) \subseteq \N$.
	\end{lemma}
	\begin{proof}
	\begin{itemize}
    \item[$\Leftarrow:$] Suppose there exists a map $G$ and some state $\rho \in [e]^{\text{diag}}_{H}$ such that
    \begin{equation}
    	G(\rho) \in [e']_{H}.
    \end{equation}
    If $G[\N] \subseteq \N$, then for any other state $\rho' \in [e]^{\text{diag}}_{H}$,
    \begin{equation}
    \begin{split}
        \mc E(G(\rho')) &= \mc E(G(\rho)) + \mc E(G(N)) \\
        &= e' + \mc E(N) \\
        &= e',
    \end{split}
    \end{equation}
    and hence $G$ satisfies condition \ref{item:cond_on_gpmaps}.
    \item[$\Rightarrow:$]
Suppose that $G \in \mc G_e$. Then, for any $\rho, \rho' \in [e]^{\text{diag}}_{H}$, by \eqref{eq:expand}
\begin{align}
	\rho - \rho' &= N ,\\
	G(\rho) - G(\rho') &= N',
\end{align}
and hence, by the linearity of CPTP-maps
\begin{align}
    G(N) &= G(\rho - \rho') \\
    	&= G(\rho) - G(\rho') \nonumber \\
    	&= N'.\nonumber
\end{align}
This implies that $G[\N_{e}] \subseteq \N$, where
\begin{equation}
	\N_{e} = \{N \in \N | \exists \rho, \rho' \in [e]^{\text{diag}}_{H}: \rho + N = \rho'\}.
\end{equation}
To expand this to the whole of $\N$, note that for non-trivial equivalence classes, $\mc G_e = \mc G_{e'}$ , $\N_{e}$ has the topology of the ball $B_{d-3}$, which implies that there exists a complete $(d-2)$-dimensional basis $\{N_i\}$ of $\N_{e}$. Moreover, this basis also constitutes a basis for $\N$. Hence, for non-trivial macrostates, $G[\N_e] \subseteq \N$ implies $G[\N] \subseteq \N$.
\end{itemize}
	\end{proof}

Note that the above proof only works, if the initial equivalence class has more than one member since otherwise $\N_{e}$ consists only of the null-vector and does not have the required topological structure. In the remainder, we therefore assume that the initial energies correspond to non-trivial equivalence classes.

A corollary of Lemma \ref{lem:nton} is that the set of macrostate GP-maps is the same, regardless of the initial energy, $\mc G_e = \mc G_{e'}$. This allows us to drop the index in the following. Then, by \eqref{eq:expand} we have
\begin{align}
    \max_{(e,H) \overset{\beta\text{-mGP}}{\rightarrow} (e',H)} \!\!\!\!\!\!\!\!e' &= \max_{G \in \mc G} \mc E(G(\rho)), \rho \in [e]^{\text{diag}}_H \\
    &= \max_{G \in \mc G} \mc E(\gamma_{e}(H) + \alpha(e) G(H - T) + G(N(\rho))) \nonumber \\
    &= e_\beta(H) + \max_{G \in \mc G} \alpha(e) \mc E(G(H - T)).\nonumber
\end{align}
Finally, note that
\begin{eqnarray}
\max_{G \in \mc G} \alpha(e) \mc E(G(H - T)) &=\nonumber \\
\begin{cases}
	\alpha(e) \max_{g \in \mc G} \mc E(G(H - T)), &\text{ if } e \geq e_\beta(H), \\
	\alpha(e) \min_{g \in \mc G} \mc E(G(H - T)), &\text{ if } e < e_\beta(H)	,
\end{cases}
\end{eqnarray}
because $\alpha(e)$ flips sign around $e_\beta(H)$. Defining the constants
\begin{align}
F_{\beta, H} &= \max_{g \in \mc G} \mc E(G(H - T)) ,\\
K_{\beta, H} &= \min_{g \in \mc G} \mc E(G(H - T)),
\end{align}
we then have
\begin{align}\label{eq:max_prob}
    \max_{(e,H) \overset{\beta\text{-mGP}}{\rightarrow} (e',H)} \!\!\!\!\!\!\!\!e' =
    \begin{cases}
	e_\beta(H) + \alpha(e) F_{\beta, H}, &\text{ if } e \geq e_\beta(H) ,\\
	e_\beta(H) + \alpha(e) K_{\beta, H}, &\text{ if } e < e_\beta(H).
\end{cases}
\end{align}
Similarly,
\begin{align} \label{eq:min_prob}
\min_{(e,H) \overset{\beta\text{-mGP}}{\rightarrow} (e',H)}\!\!\!\!\!\!\!\! e' &=
\begin{cases}
	 e_\beta(H) + \alpha(e) K_{\beta, H} , \text{ if } e \geq e_\beta(H), \\
	e_\beta(H) + \alpha(e) F_{\beta, H}, \text{ if } e < e_\beta(H).\\
\end{cases}
\end{align}
In the final step, we will now discuss the values of $F_{\beta,H}$ and $K_{\beta,H}$. The former can be found analytically to be such that
\begin{align}
    e_\beta(H) + \alpha(e) F_{\beta, H} = e. \label{eq:app:upper}
\end{align}
To see this, note that the upper term in \eqref{eq:max_prob} denotes the maximum reachable energy if the initial energy lies \emph{above} the thermal energy (see Fig.\ \ref{fig:comp}). This is trivially is at least $e$ (because the identity is always a macrostate GP-map). Now, if it was the case that
\begin{align}
    e_\beta(H) + \alpha(e) F_{\beta, H} > e, \label{eq:above_e}
\end{align}
then this would imply that there exists a GP-map $G$ such that
\begin{align}
    \mc E(G(\gamma_{e}(H))) > e.
\end{align}
In this case, $G$ would have certainly increased the free energy $\Delta F(\rho) := S(\rho || \gamma_\beta(H))$ of the system, by monotonicity of the free energy of thermal states in $e$: For any $e' > e, \rho \in [e']_H$,
\begin{align}
   \Delta F(\gamma_{e}(H)) & < \Delta F(\gamma_{\beta_S(e')}(H)) &\leq \Delta F(\rho).
\end{align}
Results from Ref.\  \cite{Janzing2000Thermodynamic} imply that no GP-map can increase the free energy of the system, so that \eqref{eq:above_e} cannot be true, and hence $F_{\beta,H}$ is determined by \eqref{eq:app:upper}.

Regarding $K_{\beta,H}$, it cannot in general be fixed analytically and depends on the $H$ and $\beta$. However, it can readily be computed with a linear program. This is because for any initial energy $e$, the optimization problems \eqref{eq:max_prob} and \eqref{eq:min_prob} can be cast as linear programs. This is true since achievable state transitions under general GP-maps can be formulated as an LP \cite{RenesCost, Janzing2000Thermodynamic}, and Lemma \ref{lem:nton} shows that the only further constraint on macrostate GP-maps is itself linear, namely that $\mc G(\mathcal{N}) \subseteq \N$. Finally, note also that a similar Lemma to Lemma \ref{lem:nton} can be shown to hold true for several commuting observables $\mc Q$. There, each of the observables $Q^j$ is bounded linearly, so that, in total, the reachable states will be characterized by piece-wise linear bounds, instead of a single linear bound. Since this lemma is a straightforward generalization of Lemma \ref{lem:linearity}, we omit its proof here.

\section{Macrostate and commuting macrostate operations in the macroscopic limit}  \label{app:macro}

In this section we discuss the value of the higher moments of the energy difference $X$ when performing a macrostate operation.  As stated in the main text, we assume that $H=\sum_i H^i$. We first consider the case of a system whose subsystems are uncorrelated. That is, we assume the initial system macrostate to be of the form $(e,H)=\otimes_{i=1}^N (e_i,H^i)$. The canonical ensemble state for $(e,H)$ is
\begin{equation}
\gamma_e(H)=\bigotimes_{i=1}^N\gamma_{\frac{e}{N}}(H^i).
\end{equation}
Finally, we consider a macrostate transition $(e,H) \overset{\beta-\text{mac}}{\rightarrow} \rho_f$, where we also assume that
\begin{equation}
	\rho_f = \bigotimes_{i=1}^N \rho_f^i.
\end{equation}
We are interested in the distribution $P(X)$, where $X$ is the change in energy under this macrostate transition.

To see that $P$ will be normally distributed, we implement the above transition by acting on each of the subsystems independently. By Theorem~\ref{thm:main}, we know that this is possible. In particular, by the procedure in Appendix~\ref{app:main}, we can implement the transition 
\begin{align}
    (e_i,H^i) \overset{\beta-\text{mac}}{\rightarrow} \gamma_{e}(H^i)
\end{align} as a macrostate transition, for any subsystem $i$. This produces a change in energy $X_i$ with mean $\mu_i$ and variance $\sigma_i^2$, which is finite for bounded $H_i$. Let $s_N^2 = \sum_i^N \sigma_i^2$. Then, by the \emph{Lyapunov Central Limit Theorem},
we have that the total change in energy, $X = \sum_i X_i$, converges in distribution to a normal distribution,
\begin{align} \label{eq:clt}
    \lim_{N \to \infty} X \overset{d}{\to} \mc N(\sum_i \mu_i = e'-e, s_N^2),
\end{align}
with $e'$ being the final energy of the system, if the following condition is satisfied: There exists a $\delta > 0$ such that
\begin{align}
    \lim_{N \to \infty} \frac{1}{s_N^{2+ \delta}} \sum_i^N \mathbb{E}[|X_i - \mu_i|^{2 + \delta}] = 0.
\end{align}
Choosing $\delta = 1$ and since $s_N^2 = O(N)$, this is satisfied if $\sum_i^N \mathbb{E}[|X_i - \mu_i|^{2 + \delta}]  = O(N)$. This is a physically reasonable assumption to make.  Now, from \eqref{eq:clt} it follows that the energy change \emph{per subsystem} is normally distributed as
\begin{align} \label{eq:clt2}
    \lim_{N \to \infty}\frac{X}{N} \overset{d}{\to} \mc N(e'-e, \frac{s_N^2}{N}).
\end{align}
In terms of the higher moments this means the following. Let
\begin{align}
    \mu_n(X) := \mathbb{E}[(X - \mu)^n], \quad n \in {1, 2, \dots}
\end{align}
be the moments of a random variable $X$. If this $X$ is normally distributed with variance $\sigma^2$, then independent of its mean the following is true and can be verified by evaluation.
\begin{align} \label{eq:moments_prop}
    \mu_{2n}(X) = \sigma^{2n}(2n -1)!! ,
    & \quad \mu_{2n +1}(Y) = 0.
\end{align}
Combining this with \eqref{eq:clt2} we find that the higher moments \emph{per subsystem} vanish in the macroscopic limit:
 \begin{align}
       \lim_{N \to \infty}  \mu_{2n}(X/N)&= \lim_{N \to \infty}   \left(\frac{s_N}{\sqrt{N}}\right)^{2n}(2n -1)!! & = 0.
 \end{align}
As stated in the main text, this can be seen as an argument in favour of the assignment of the ensemble to macrostates, for large weakly-correlated systems, as long as one tolerates violations of \eqref{eq:exactcommutation} -- as measured by the higher moments
-- that are negligible in comparison with the typical energy scales involved in the thermodynamic operation. Of course, a similar argument can be made for the case of weakly correlated systems. However, for conceptual clarity we here restricted to the independent case.rational equivalence is regained for those subsystems.

\end{document}